\documentclass[journal]{IEEEtaes}

\ifCLASSOPTIONcompsoc
\usepackage[caption=false,font=normalsize,labelfont=sf,textfont=sf]{subfig}
    
\else
\usepackage[caption=false,font=footnotesize]{subfig}
\fi
\usepackage{color}
\definecolor{deeppink}{rgb}{1.0, 0.08, 0.58}
\usepackage{easyReview}
\usepackage{color-edits}
\addauthor{JW}{deeppink}
\usepackage{setspace}
\usepackage{graphicx}
\usepackage{mathrsfs}
\usepackage{amsmath}
\usepackage{amsfonts}
\usepackage{multirow}
\usepackage{amsthm}
\usepackage{booktabs}
\usepackage{verbatim}
\usepackage{lineno}
\usepackage{hyperref}  %

\usepackage{marginnote}

\newtheorem{theorem}{Theorem}  
\labelformat{subsection}{\thesection-#1}
\begin{document}
\title{Self-Supervised-ISAR-Net Enables Fast Sparse ISAR Imaging}
\author{Ziwen~Wang}
\member{Graduate student Member,~IEEE}

\author{Jianping~Wang}
\member{Member,~IEEE}

\author{Pucheng~Li}
\member{Member,~IEEE}

\author{Yifan~Wu}
\member{Graduate student Member,~IEEE}

\author{Zegang~Ding}
\member{Senior Member,~IEEE}
\affil{Beijing Institute of Technology}

\receiveddate{Manuscript received XXXXX 00, 0000.\\
This work was supported in part by the National Science Foundation of China under Grant 62227901 and Grant 62101035, the Key Program of the National Science Foundation of China under Grant 61931002. }
\corresp{ {\itshape (Corresponding author: Jianping Wang)}.}

\authoraddress{Ziwen Wang, Jianping Wang, Pucheng Li, Yifan Wu, and Zegang Ding are with the School of Information and Electronics, Beijing Institute of Technology, and also with the Key Laboratory of Electronic and Information Technology in Satellite Navigation, Ministry of Education, Beijing 100081, China.
(e-mail: \href{mailto:3120235592@bit.edu.cn}{3120235592@bit.edu.cn}; \href{mailto:jianpingwang@bit.edu.cn}{jianpingwang@bit.edu.cn}; \href{mailto:pucklee1111@163.com}{pucklee1111@163.com}; 
\href{mailto:1540830052@qq.com}{1540830052@qq.com};
\href{mailto:z.ding@bit.edu.cn}{z.ding@bit.edu.cn}).}

\markboth{Journal of \LaTeX\ Class Files,~Vol.~14, No.~8, August~2021}%
{Shell \MakeLowercase{\textit{et al.}}: A Sample Article Using IEEEtran.cls for IEEE Journals}


\maketitle

\begin{abstract}
Numerous sparse inverse synthetic aperture radar (ISAR) imaging methods based on unfolded neural networks have been developed for high-quality image reconstruction with sparse measurements. However, their training typically requires paired ISAR images and echoes, which are often difficult to obtain.
Meanwhile, one property can be observed that for a certain sparse measurement configuration of ISAR, when a target is rotated around its center of mass, only the image of the target undergoes the corresponding rotation after ISAR imaging, while the grating lobes do not follow this rotation and are solely determined by the sparse-sampling pattern. This property is mathematically termed as the equivariant property.
Taking advantage of this property, an unfolded neural network for sparse ISAR imaging with self-supervised learning, named SS-ISAR-Net is proposed.
It effectively mitigates grating lobes caused by sparse radar echo, allowing high-quality training to be achieved using only sparse radar echo data. The superiority of the proposed SS-ISAR-Net, compared to existing methods, is verified through experiments with both synthetic and real-world measurement data.
\end{abstract}

\begin{IEEEkeywords}
  Sparse ISAR imaging, Unfolded network, Equivariant constraint, Sparse radar echo
\end{IEEEkeywords}

\section{Introduction}
\IEEEPARstart{I}{nverse} synthetic aperture radar (ISAR) plays a crucial role in space situation awareness and target monitoring due to its all-day, all-weather imaging capability~\cite{chan2008introduction,samczynski2016sar}. It obtains high range resolution by employing wideband signals and high azimuth resolution by exploiting the relative motion between the radar and the target. The rotational component of the relative motion serves as the primary source of azimuth resolution, whereas the translational component does not contribute to imaging quality and must be compensated. 
Following translational motion compensation~\cite{li2023joint, yang2020integration, zhang2024joint}, a widely used and effective signal processing approach for ISAR imaging is the range-Doppler (RD) method~\cite{prickett1980principles}. In practical ISAR imaging scenarios, interferences, as well as radar system resource scheduling, may hinder the acquisition of complete echoes in both the azimuth and range dimensions according to the Nyquist sampling theorem~\cite{5783920,10634748}. The process of reconstructing an ISAR image from incomplete radar echo data is referred to as sparse ISAR imaging.
However, for sparse ISAR imaging, RD is not applicable as it would lead to the emergence of grating lobes that cause ambiguities. 

To overcome this problem, many sparse ISAR imaging methods based on compressed sensing (CS) and deep learning (DL) techniques are developed~\cite{donoho2006compressed,8550778}. They are discussed in detail in the following.
      


\textbf{Compressed sensing (CS)-based ISAR imaging methods}: 
In sparse ISAR imaging, CS-based optimization methods are primarily categorized into three types: greedy algorithms, sparse Bayesian algorithms, and convex optimization algorithms based on sparsity constraint. 

The basic principle of greedy algorithms is to iteratively select the locally optimal matching element, compute the signal residual, and continue the process until the reconstruction error falls below a predefined threshold ~\cite{mallat1993matching,pati1993orthogonal}.
Greedy algorithm-based methods offer a simple and intuitive strategy for ISAR imaging. 
However, their imaging accuracy deteriorates significantly when the sparsity of targets is unknown — a condition that is typically not available in practical scenarios~\cite{cheng2019fast}.

Sparse Bayesian Learning (SBL) algorithms assume that each element of the reconstructed signal follows a predefined sparse prior distribution. Using Bayes' theorem, these algorithms maximize the posterior probability to obtain an optimal estimation~\cite{wipf2004sparse}. Due to their probabilistic framework, SBL algorithms exhibit strong modeling capabilities and facilitate the acquisition of sparse solutions more effectively. 
However, SBL methods require matrix inversion, which results in high computational complexity. Although some efforts have been made to address this issue~\cite{zhang2024joint}, SBL algorithms still suffer from hyperparameter sensitivity and sparse model mismatch. 


Convex optimization-based methods address the (relaxed) convex inversion problem of the sparse ISAR imaging through iterative solvers, including the fast iterative shrinkage threshold algorithm (FISTA)~\cite{beck2009fast}, and the alternating direction method of multipliers (ADMM)~\cite{boyd2011distributed, wang20223}. Among these methods, ADMM has attracted significant attention in sparse ISAR imaging thanks to its ability to tackle multivariate optimization problems~\cite{9130034, zhang2020computationally}. 
However, the selection of related hyperparameters in ADMM-based ISAR imaging methods is heavily based on empirical tuning, which generally is not trivial and could significantly affect imaging accuracy under certain circumstances. Moreover, since ISAR targets often exhibit continuous structural characteristics, enforcing a regular sparse prior may lead to reconstruction mismatches. 
To address these challenges, several unfolded neural networks based on convex optimization methods have been introduced as a promising alternative~\cite{10038885,li2022high}. 

\textbf{Deep learning (DL)-based ISAR imaging methods}: 
DL-based ISAR imaging methods include purely data-driven DL methods and model-guided DL methods. 
For the purely data-driven methods, an end-to-end deep neural network is trained directly by taking radar echo data or pre-processed images as input and the corresponding ISAR images as output~\cite{wang2023deep}. Although good results can be obtained, these methods generally suffer from a lack of interpretability and poor generalization performance. 
      
To address these limitations, model-guided DL methods for ISAR imaging are developed by unfolding the optimization algorithms mentioned above into an interpretable iterative neural network, which are computationally efficient and achieve higher imaging accuracy compared to conventional optimization methods~\cite{wang2022efficient, huang2025cv, pu2022sae}. Meanwhile, an advanced technique known as SwinRL-ADMM integrates reinforcement learning with the Swin Transformer to adaptively learn iterative hyperparameters, with the aim of achieving optimal performance~\cite{li2025tuning}. Despite these advancements, training the aforementioned end-to-end neural networks requires paired radar echoes and corresponding images. However, in the field of ISAR imaging, acquiring a sufficiently large training dataset remains a significant challenge.

Therefore, developing a self-supervised training strategy for DL-based ISAR imaging methods is crucial by exploiting target feature prior information from sparse radar echo data, which usually requires finding some kind of inherent constraint in the ISAR imaging process \cite{gui2024survey}. In practice, an equivariance property can be observed (see Fig.~\ref{fig:EC_model}): when a target of ISAR undergoes a geometric transformation (e.g., a rotation around its center of mass ${{\bf{T}}_g}$), its true scattering features rotate correspondingly in the reconstructed ISAR image. In contrast, grating lobes introduced by the sparse imaging operator ${{\bf{F}}_s}$ do not follow this transformation, since they are caused by aliasing effects introduced by the imaging operator. This indicates that only the target features exhibit equivariance whereas the grating lobes do not. This property provides the mathematical basis for grating lobes suppression via the equivariant constraint.

Building on the equivariant constraint property, in this paper, an unfolded neural network for sparse ISAR imaging with self-supervised learning is proposed, referred to as SS-ISAR-Net, which allows the ISAR imaging network to be trained using only incomplete radar echo data. The proposed method achieves comparable or even better sparse ISAR imaging performance in contrast to state-of-the-art supervised DL methods.

The main contributions of this paper are as follows:
\begin{enumerate}
  \item An ADMM-based unfolded ISAR imaging network with a local feature-adaptive threshold mechanism is proposed, enabling the learning of hidden local structural features from data.
  \item A self-supervised learning framework with an equivariant constraint is proposed and demonstrated that imposing the equivariant constraint is mathematically equivalent to extending the measurement dimension of the ISAR sensing matrix.
  \item A Noise2Noise (N2N)-based echo self-denoising module is introduced to mitigate the impact of noise on the proposed SS-ISAR-Net.
  \item A series of experiments were conducted to validate the effectiveness of the proposed SS-ISAR-Net.
  \end{enumerate}


\begin{figure}[t]
  \centering
  \includegraphics[width =7.6cm, height =7cm]{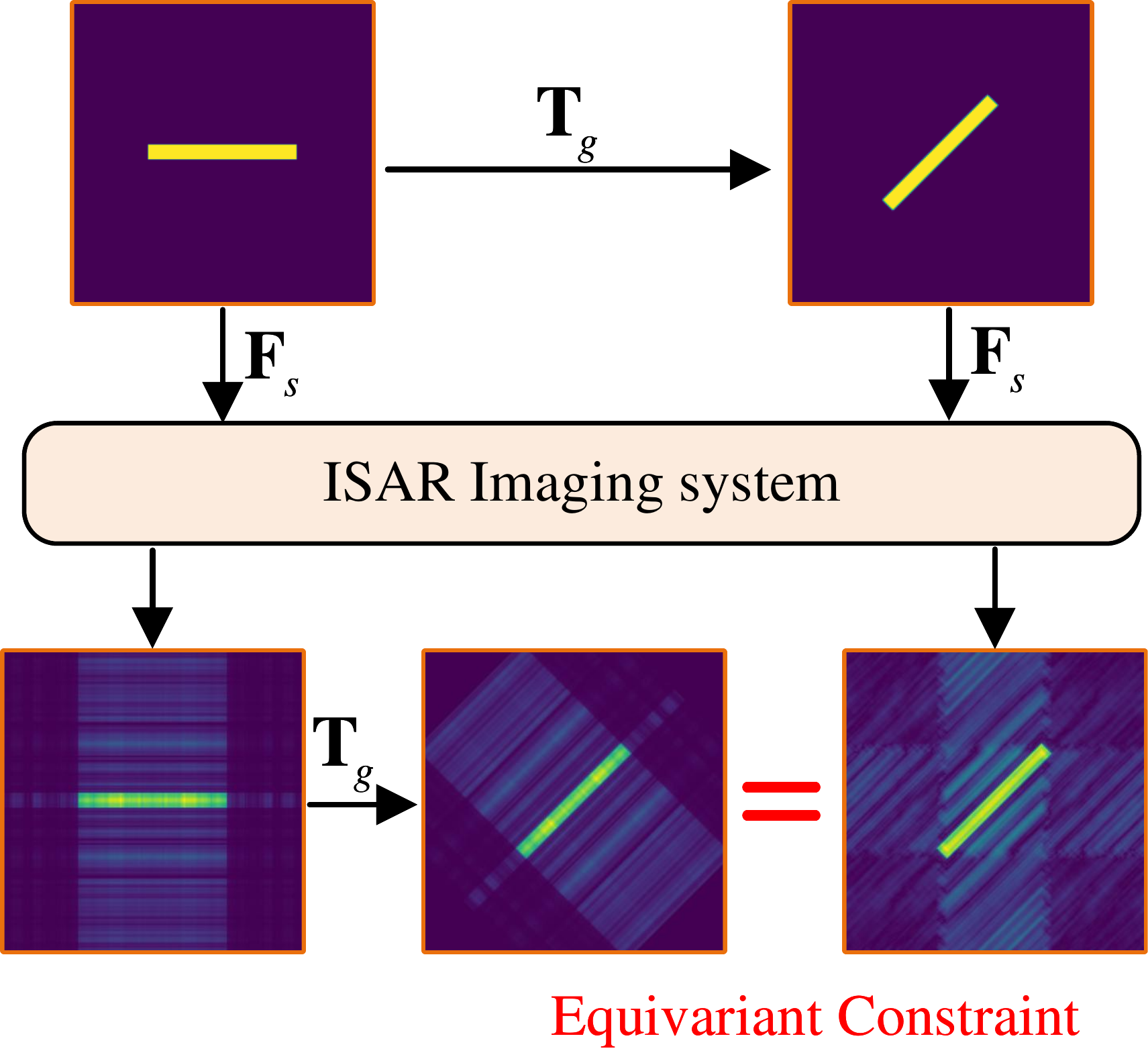}
  \caption{Illustration of the equivariant constraint in ISAR imaging}\label{fig:EC_model}
\end{figure}

The rest of the paper is organized as follows: Section~\ref{Sec:2} presents the ISAR signal model and a sparse ISAR imaging solver based on ADMM. In Section~\ref{Sec:3}, an ADMM unfolded ISAR imaging network is proposed, achieving an end-to-end self-supervised training by taking advantage of an equivariant constraint. In Section~\ref{Sec:4}, the results of numerical simulations and real-data experiments are presented to illustrate the effectiveness of the proposed SS-ISAR-Net. In Section~\ref{Sec:5}, the impact of some  hyperparameters of the network on the imaging performance is discussed. Finally, conclusions and further plans are given in Section~\ref{Sec:6}.

\section{ISAR Signal Model and Optimization Solver}\label{Sec:2}
In this section, a general ISAR imaging model is introduced first, and then an ADMM-based framework for sparse ISAR imaging is provided. 
        
\subsection{ISAR imaging model}\label{Sec:2-1}

\begin{figure}[!t]
  \hspace{-0.2cm}
  \includegraphics[width =7cm, height =6.3cm]{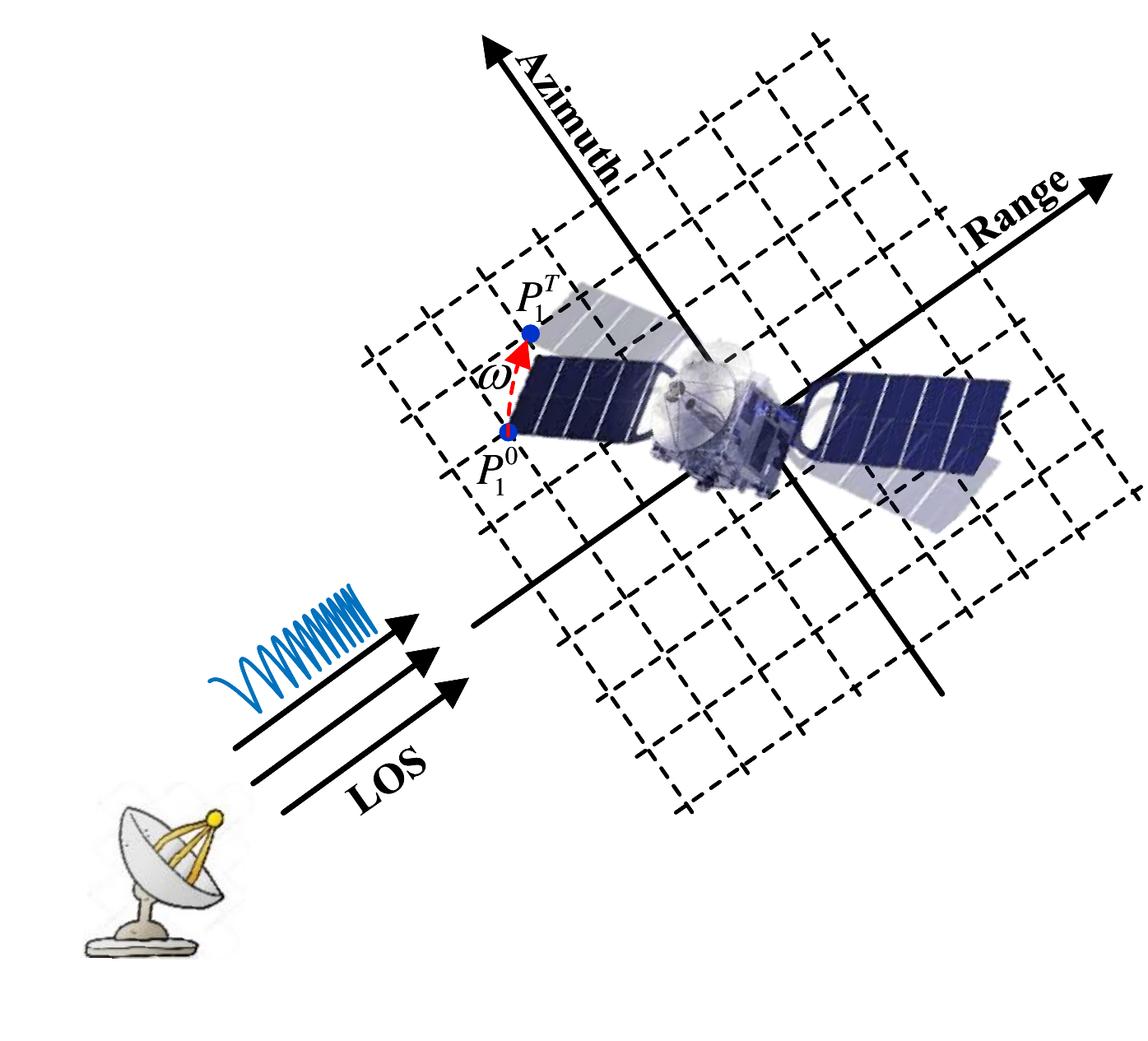}
  \vspace{-0.5cm}
  \caption{Illustration of the ISAR imaging model}\label{fig:ISAR}
\end{figure}

A general ISAR imaging model is shown in Fig.~\ref{fig:ISAR}. ISAR obtains high-resolution images of targets in both range and azimuth dimensions by transmitting wideband signals and utilizing the target's relative azimuthal rotation with respect to the radar, respectively. Its acquired echoes of targets after demodulation and range-matched filtering can be expressed as
\begin{align}\label{eq:echo1}
  s({f_r},{t_m}) = \sum\limits_{k = 1}^K &{{A_k}} {\text{rect}}\left( {\frac{{{f_r}}}{B}} \right) \notag \\
   &\cdot \exp \left( { - j\frac{{4\pi }}{c}\left( {{f_c} + {f_r}} \right) {R_k}\left( {{t_m}} \right)} \right) 
\end{align} 
where $f_r$ denotes the frequency counterpart of the fast time, and $f_c$ represents the center frequency. $t_m$ is the slow time in azimuth. The variable $k = [1,\,2,\, \cdots,\, K]$ represents the index of scattering points, while $A_k$ and {$ R_k\left( t_m \right)$} are the scattering amplitude and instantaneous slant range at $t_m$ of the $k^\text{th}$ point scatterer, respectively.
Considering the assumption of small-angle rotation of the target, 
the instantaneous slant range after compensating the translational component~\cite{ding2020parametric} can be expressed as
\begin{equation}
   {R_k}\left( {{t_m}} \right) \approx {x_k} + {y_k}\omega {t_m}
  \label{eq:rotation}
\end{equation}
where $x_k$ and $y_k$ denote range and azimuth coordinates of the $k^\text{th}$ target point, respectively, while $\omega$ represents the relative angular rotation speed. By combining~\eqref{eq:echo1} and~\eqref{eq:rotation}, neglecting the constant phase term and compensating for migration through range cells~\cite{9130034,xing2004migration}, the resultant signal can be given by
\begin{align}
s\left( {{f_r},{t_m}} \right) = \sum\limits_{k = 1}^K {{A_k}} &{\rm{rect}}\left( {\frac{{{f_r}}}{B}} \right)\cdot\exp \left( { - j\frac{{4\pi }}{c}\left( {{f_c} + {f_r}} \right){x_k}} \right)\notag\\
 &\qquad\cdot \exp \left( { - j\frac{{4\pi {f_c}}}{c}{y_k}\omega {t_m}} \right)
   \label{eq:echo2}
\end{align} 

Discretizing the frequency in range and slow time in azimuth, we have ${f_r} = n\Delta f$, ${t_m} =m/\mathrm{PRF}$, where $n \in \left[ {1,N} \right],\, m \in \left[ {1,M} \right]$. Here, $\Delta f$ and $\text{PRF}$ are the frequency interval and azimuth sampling rate, respectively. Accounting for all the measurements and possible measurement errors and noise, one can rearrange~\eqref{eq:echo2} in a matrix form
\begin{equation}
  {\mathbf{Y}} = {\mathbf{AXB}} + {\mathbf{N}} 
  \label{eq:echo matrix}
\end{equation}
where ${\mathbf{Y}} \in {{\mathbb{C} }^{N \times M}}$ and ${\mathbf{N}} \in {{\mathbb{C} }^{N \times M}}$ represent the ISAR echo matrix and noise matrix, respectively. The ISAR image to be reconstructed can be represented as a 2D matrix ${\mathbf{X}} \in {{\mathbb{C} }^{P \times Q}}$, where ${{\bf{X}}_{p,q}} = {A_{p,q}}$ represents the complex amplitude at the pixel $\left( {p,q} \right)$. ${\mathbf{A}} \in {{\mathbb{C} }^{N \times P}}$ and ${\mathbf{B}} \in {{\mathbb{C} }^{Q \times M}}$ denote the forward matrices for range and azimuth information acquisition respectively, which are given by 
\begin{equation}
  \begin{gathered}
  A_{n,p} = \exp \left( { - j\frac{{4\pi }}{c}\left( {{f_c} + n\Delta f} \right){x_p}} \right) \hfill \\
  B_{q,m} = \exp \left( { - j\frac{{4\pi {f_c}}}{c}\frac{{{y_q}\omega m}}{{{\text{PRF}}}}} \right) \hfill \\ 
\end{gathered}   
\end{equation}
where $\left( {{x_p},{y_q}} \right)$ represents the $\left( {p,q} \right)$-th entry of the ISAR image ${\mathbf{X}}$.
In practice, the continuous observation of the target in both the range frequency and azimuth time may be affected due to interferences, as well as radar system resource scheduling constraints. Consequently, sparse sampling has to be taken in both the range and azimuth dimensions, which can be expressed as
\begin{equation}
{{\mathbf{Y}}_s} = {{\mathbf{P}}_a}\left( {{\mathbf{AXB}} + {\mathbf{N}}} \right){{\mathbf{P}}_b} = {{\mathbf{A}}_s}{\mathbf{X}}{{\mathbf{B}}_s} + {{\mathbf{N}}_s}
\label{eq:AXB}
\end{equation}
where ${{\mathbf{Y}}_s}\in {{\mathbb{C} }^{N_s \times M_s}}$ represents the downsampled radar echo. 
${{\mathbf{P}}_a}\in {{\mathbb{R} }^{N_s \times N}}$ and ${{\mathbf{P}}_b}\in {{\mathbb{R} }^{M \times M_s}}$ denote the possible sparse sampling matrix in range and azimuth, respectively, which indicate that only $N_s$ valid range frequency points and $M_s$ valid azimuth time samples are acquired from $N$ range frequency samples and $M$ azimuth time samples. ${{\mathbf{A}}_s} \in {{\mathbb{C} }^{{N_s} \times P}}$ and ${{\mathbf{B}}_s} \in {{\mathbb{C} }^{Q \times {M_s}}}$ serve as the forward downsampled operators for range and azimuth. Additionally, ${{\mathbf{N}}_s}\in {{\mathbb{C} }^{N_s \times M_s}}$ represents the corresponding downsampled noise. As a result, the radar down-sampling rate can be defined as 
\begin{equation}
  \gamma  = \frac{{{N_s}{M_s}}}{{NM}} \times 100\% 
  \label{eq:sampling}
\end{equation}

As $\gamma$ is generally smaller than $100\%$, to recover $\mathbf{X}$ based on \eqref{eq:AXB} often leads to the sparse ISAR imaging.
Below, we briefly introduce a widely adopted sparse ISAR imaging method based on ADMM.

\subsection{Sparse ISAR imaging Solver based on ADMM}\label{Sec:2-2}
Sparse ISAR imaging methods typically achieve superior performance by effectively utilizing target prior knowledge and optimization solvers. Among various optimization techniques, ADMM has emerged as a widely adopted and computationally efficient approach for sparse ISAR image reconstruction~\cite{boyd2011distributed,9130034}. Specifically, It formulates the optimization problem as the sum of measurement constraints, as represented in~\eqref{eq:AXB}, combined with the prior regularization to enhance reconstruction accuracy.
    
\begin{equation}
  \mathop {\arg \min }\limits_{\mathbf{X}} \frac{1}{2}\left\| {{{\mathbf{Y}}_s} - {{\mathbf{A}}_s}{\mathbf{X}}{{\mathbf{B}}_s}} \right\|_F^2 + \lambda \mathcal{R}\left( {\mathbf{X}} \right) 
  \label{eq:step1}
\end{equation}
where ${\left\|  \cdot  \right\|_F}$ denotes the Frobenius norm of a matrix. $\mathcal{R}\left( {\mathbf{X}} \right)$ is the prior regularization term, while $\lambda$ controls the regularization strength. Since solving these two terms simultaneously is challenging, ADMM introduces an auxiliary variable ${\mathbf{Z}}$, allowing~\eqref{eq:step1} to be rewritten as
\begin{equation}
\mathop {\arg \min }\limits_{{\mathbf{X}},{\mathbf{Z}}} \frac{1}{2}\left\| {{{\mathbf{Y}}_s} - {{\mathbf{A}}_s}{\mathbf{X}}{{\mathbf{B}}_s}} \right\|_F^2 + \lambda \mathcal{R}\left( {\mathbf{Z}} \right)\quad{\text{s.t.}}\quad{\mathbf{X}} = {\mathbf{Z}}
\label{eq:step2}
\end{equation}
The constrained optimization problem in \eqref{eq:step2} can be further rewritten as an unconstrained one by using its corresponding augmented Lagrangian function
\begin{align}\label{eq:step3}\notag
  \hspace{-0.2cm}
  \mathop {\arg \min }\limits_{\mathbf{X},\mathbf{Z}} &\frac{1}{2}\left\| \mathbf{Y}_s - \mathbf{A}_s \mathbf{X} \mathbf{B}_s \right\|_F^2 + \lambda \mathcal{R}\left( {\mathbf{Z}} \right)  \\
  &\qquad\qquad+\left\langle \mathbf{\Phi}, \mathbf{X} - \mathbf{Z} \right\rangle 
   + \frac{\rho }{2}\left\| {{\mathbf{X}} - {\mathbf{Z}}} \right\|_F^2  
\end{align} 
where $\mathbf{\Phi}$ is the Lagrangian multiplier, and $\rho$ is the penalty coefficient. By applying a scaling transformation and introducing a dual variable, i.e., $\mathbf{U} = \mathbf{\Phi} / \rho$,~\eqref{eq:step3} can be reformulated as 
\begin{equation}\label{eq:step4}
  \mathop {\arg \min }\limits_{{\mathbf{X}},{\mathbf{Z}}} \frac{1}{2}\left\| {{{\mathbf{Y}}_s} - {{\mathbf{A}}_s}{\mathbf{X}}{{\mathbf{B}}_s}} \right\|_F^2 + \lambda \mathcal{R}\left( {\mathbf{Z}} \right){\text{ +  }}\frac{\rho }{2}\left\| {{\mathbf{X}} - {\mathbf{Z}} + {\mathbf{U}}} \right\|_F^2
\end{equation}

The optimization problem in~\eqref{eq:step4} can be solved by iteratively optimizing three optimization problems of $\mathbf{X}$, $\mathbf{Z}$, and $\mathbf{U}$
\begin{align}\label{eq:diedai}
    {{\mathbf{X}}^k} &= \mathop {\arg \min }\limits_{\mathbf{X}} \frac{1}{2}\left\| {{{\mathbf{Y}}_s} - {{\mathbf{A}}_s}{\mathbf{X}}{{\mathbf{B}}_s}} \right\|_F^2{\text{  +   }}\frac{\rho }{2}\left\| {{\mathbf{X}} - {\mathbf{Z}} + {\mathbf{U}}} \right\|_F^2 \notag\\
  {{\mathbf{Z}}^k}&=\mathop {\arg \min }\limits_{\mathbf{Z}} \lambda \mathcal{R}\left( {\mathbf{Z}} \right){\text{ +  }}\frac{\rho }{2}\left\| {{\mathbf{X}} - {\mathbf{Z}} + {\mathbf{U}}} \right\|_F^2 \\
  {{\mathbf{U}}^k} &= {{\mathbf{U}}^{k - 1}} + \rho \left( {{{\mathbf{X}}^k} - {{\mathbf{Z}}^k}} \right) \notag 
\end{align}

The estimated ISAR image ${\mathbf{X}}$ can be obtained through iterative optimization of~\eqref{eq:diedai}. In ISAR imaging, sparsity is commonly used as a prior term, typically represented by the $\ell _1$ norm, that is, $\mathcal{R}\left( {\mathbf{Z}} \right) = {\left\| {\bf{Z}} \right\|_1}$. The ${\ell _1}$ norm admits a proximal operator that corresponds to the soft-thresholding function. Therefore, the second step iteration in~\eqref{eq:diedai} can be explicitly expressed as follows
\begin{equation} \label{eq:soft}
{{\mathbf{Z}}^k} = \mathcal{S}\left( {{{\mathbf{X}}^k} + {{\mathbf{U}}^{k - 1}};{\lambda / \rho }} \right)
\end{equation}
where $\mathcal{S}\left(\cdot;\,\lambda  / \rho \right):\mathbb{C}^{P \times Q} \to \mathbb{C}^{P \times Q}$ is the soft threshold function, 
$\mathcal{S}\left( \mathbf{X}; \lambda/\rho  \right) = \text{sign}\left( {\mathbf{X}} \right)\max \left( \left| {\mathbf{X}} \right| - \lambda/\rho ,0 \right)$, and $\lambda / \rho $ is the threshold. Although it effectively suppresses the grating lobes and improves image quality, it can risk degrading the original structural information of the target~\cite{li2023integrated}. Moreover, the selection of hyperparameters in~\eqref{eq:diedai} typically requires empirical tuning, which limits the adaptability of the method~\cite{8550778}. 

Unfolding the above iterative process in~\eqref{eq:diedai} into a neural network offers a potential solution to these challenges~\cite{wang2022efficient, huang2025cv, pu2022sae}. However, such approaches typically require substantial amounts of paired ISAR training data, which may often be difficult to obtain.
In addition, relying solely on the globally consistent threshold sparsity constraint in~\eqref{eq:soft} may lead to suboptimal filtering performance, as it does not account for local features of the reconstructed image.

To address these limitations, we propose an unfolded ISAR network with a local feature-adaptive threshold mechanism for the learning of hidden local structural features from data. By integrating an equivariant constraint, the network enables self-supervised training using only incomplete radar echo, thus circumventing the need for paired ISAR training data.

\section{SS-ISAR-Net: a self-supervised learning-based ISAR imaging network with the equivariant constraint}\label{Sec:3}
To adaptively learn the hyperparameters in each iteration and efficiently extract prior information from training data, this section introduces an ADMM-based unfolded ISAR imaging network, which integrates a local feature-adaptive threshold (LFAT) mechanism. Building upon this framework, an equivariant constraint is incorporated to enable self-supervised, end-to-end training using only incomplete radar echo data. Additionally, to mitigate the impact of noise, a self-denoising module based on the N2N paradigm is introduced, further enhancing the robustness of the self-supervised training.

\subsection{Unfolded ISAR Imaging Network} \label{Sec:3-1}

In this section, an unfolded ISAR imaging network (named ISAR-Net) is proposed to automatically learn the hyperparameters in each iteration from the training data based on~\eqref{eq:diedai}, while incorporating a local feature-adaptive threshold mechanism to better capture the local structural characteristics. The following presents its detailed explanation.
    
\subsubsection{\textbf{Measurement reconstruction step}} 
For the first iteration step in~\eqref{eq:diedai}, obtaining a closed-form solution is challenging. Therefore, the gradient descent (GD) method is employed to get a numerical solution
\begin{align}
  {\bf{X}}_n^k &= {\bf{X}}_{n - 1}^k - l_x \left[ \rho \left( {\bf{X}}_{n - 1}^k - {\bf{Z}}^k + {\bf{U}}^k \right) \right. \notag \\[1.5mm]
  & \quad\quad\qquad \left. + {\bf{A}}_s^H \left( {\bf{A}}_s {\bf{X}}_{n - 1}^k {\bf{B}}_s - {\bf{Y}} \right) {\bf{B}}_s^H \right] \notag\\[1.5mm]
  &= \mu {\bf{X}}_{n - 1}^k + (1 - \mu) \left( {\bf{Z}}^k - {\bf{U}}^k \right) \notag \\[1.5mm]
  & \quad\quad\qquad - l_x {\bf{A}}_s^H \left( {\bf{A}}_s {\bf{X}}_{n - 1}^k {\bf{B}}_s - {\bf{Y}} \right) {\bf{B}}_s^H, \notag \\   
& \quad\quad\qquad\qquad n=1,\,2,\,\cdots,\, N
  \end{align}
where ${\mu} = \left( {1 - {l_x}\rho } \right)$, with ${l_x}$ denoting the step size of GD in $\mathbf{X}$ updating iteration. ${\mathbf{X}}_n^k$ represents the $n^\text{th}$ GD update in the $k^\text{th}$ iteration of ADMM, and $N=5$ is taken here.     
In this step, ${\mu}$ and ${{l_x}}$ are considered as learnable variables through different iterations.

\subsubsection{\textbf{Prior regularization step}}

A fixed threshold is generally used in~\eqref{eq:soft} over an entire image in the prior regularization iteration. However, since the target typically exhibits structural characteristics, the threshold for each pixel could be influenced by its surrounding features rather than remaining fixed throughout the image~\cite{li2023integrated, 9577972}. To overcome this limitation, a local feature-adaptive threshold (LFAT) is introduced, leading to a revised iterative formulation of~\eqref{eq:soft}
\begin{equation}
  {{\mathbf{Z}}^k} = \mathcal{S}\left( {{{\mathbf{X}}^k} + {{\mathbf{U}}^{k - 1}};\mathcal{T}\left( {{{\mathbf{X}}^k} + {{\mathbf{U}}^{k - 1}}} \right)} \right)
\end{equation}
where ${\mathcal{T}\left( {{{\mathbf{X}}^k} + {{\mathbf{U}}^{k - 1}}} \right)}$ represents the LFAT, which dynamically adjusts the threshold based on the surrounding target features for each pixel. In this paper, the LFAT is implemented using a convolution-based module, specifically
\begin{equation}\label{eq:adaptive threshold}
  \mathcal{T} \left( {{{\mathbf{X}}^k} + {{\mathbf{U}}^{k - 1}}} \right) = {\mathcal{C}_2}\left( {{\mathcal{R}}\left( {{\mathcal{C}_1}\left( {{{\mathbf{X}}^k} + {{\mathbf{U}}^{k - 1}}} \right)} \right)} \right)
\end{equation}
where $\mathcal{C}_1$ and $\mathcal{C}_2$ denote two different convolutional layers, and $\mathcal{R}$ represents a ReLU activation layer. The convolution-based module excels at learning structural information within the image and can adaptively determine an appropriate threshold by learning the structural features of the ISAR image.

\subsubsection{\textbf{Dual ascent step}}  
For the third step in~\eqref{eq:diedai}, the  dual ascent is adaptively performed by setting the penalty variable ${{\rho ^k}}$ as a learnable variable in different stage $k$. 
\begin{equation}\label{eq:adaptive ruo}
  {{\mathbf{U}}^k} = {{\mathbf{U}}^{k - 1}} + {\rho ^k}\left( {{{\mathbf{X}}^k} - {{\mathbf{Z}}^k}} \right)
\end{equation}
 
Thus far, the ISAR-Net has been formulated and can be trained with paired radar ISAR data to achieve high-quality ISAR image reconstruction. However, obtaining such paired datasets in the ISAR domain remains a significant challenge. To address this limitation, a self-supervised training framework that relies solely on incomplete radar echo data is suggested below.

\subsection{Self-supervised learning based on the equivariant constraint for ISAR-Net}\label{Sec:3-2}
In this section, a rotating equivariant constraint is introduced into the previously proposed ISAR-Net to achieve self-supervised learning. It equivalently expands the representation space of the ISAR forward operator. From the perspective of inversion, it injects complementary information on top of incomplete radar echo measurements, mitigating the possible aliasing in sparse ISAR imaging.

To facilitate understanding of the rotating equivariant constraint in ISAR imaging, the model in~\eqref{eq:AXB} is simplified to ${{\mathbf{Y}}_s} = {{\mathbf{F}}_s}{\mathbf{X}}$, where ${{\mathbf{F}}_s}$ can be interpreted as a forward matrix integrating both azimuth and range. Noise-related issues will be solved in Sec~\ref{Sec:3-3}. The incompleteness of ${{\mathbf{F}}_s}$ introduces a null space, thereby expanding the solution space and leading to the following suboptimal solution.
\begin{equation}
{\mathbf{\hat X}} = {\mathbf{F}}_s^\dag {{\mathbf{F}}_s}{\mathbf{X}} + {{\mathbf{N}}_{{{\mathbf{F}}_s}}} = {\mathbf{X}} + {{\mathbf{N}}_{{{\mathbf{F}}_s}}}
  \label{subsolution}
\end{equation}
here, $\mathbf{F}_s^\dag $ represents the pseudo-inverse of $\mathbf{F}_s$, while ${{\mathbf{N}}_{{{\mathbf{F}}_s}}} \in \mathcal{N}\left( {{{\mathbf{F}}_s}} \right)$ is an element of the null space of ${\mathbf{F}}_s$, i.e., $\mathcal{N}\left( {{{\mathbf{F}}_s}} \right)$, satisfying ${{\mathbf{F}}_s}{{\mathbf{N}}_{{{\mathbf{F}}_s}}} = \emptyset$. In ISAR imaging systems, existence of a non-trivial null space of the forward sensing operator $\mathbf{F}_s$ primarily manifests as pronounced grating lobes that significantly degrade image quality. While the suboptimal solution in~\eqref{subsolution} satisfies the measurement constraints in~\eqref{eq:AXB}, it cannot guarantee to be an accurate ISAR image reconstruction. To address this fundamental limitation, to reduce the null space of the forward operator ${\mathbf{F}}_s$, thereby suppressing grating lobes artifacts, is a feasible solution. The equivariant constraint introduced in the following provides a mathematically rigorous framework to achieve this objective.

\begin{figure}[t]
    \centering
  \includegraphics[width =4cm, height =4cm]{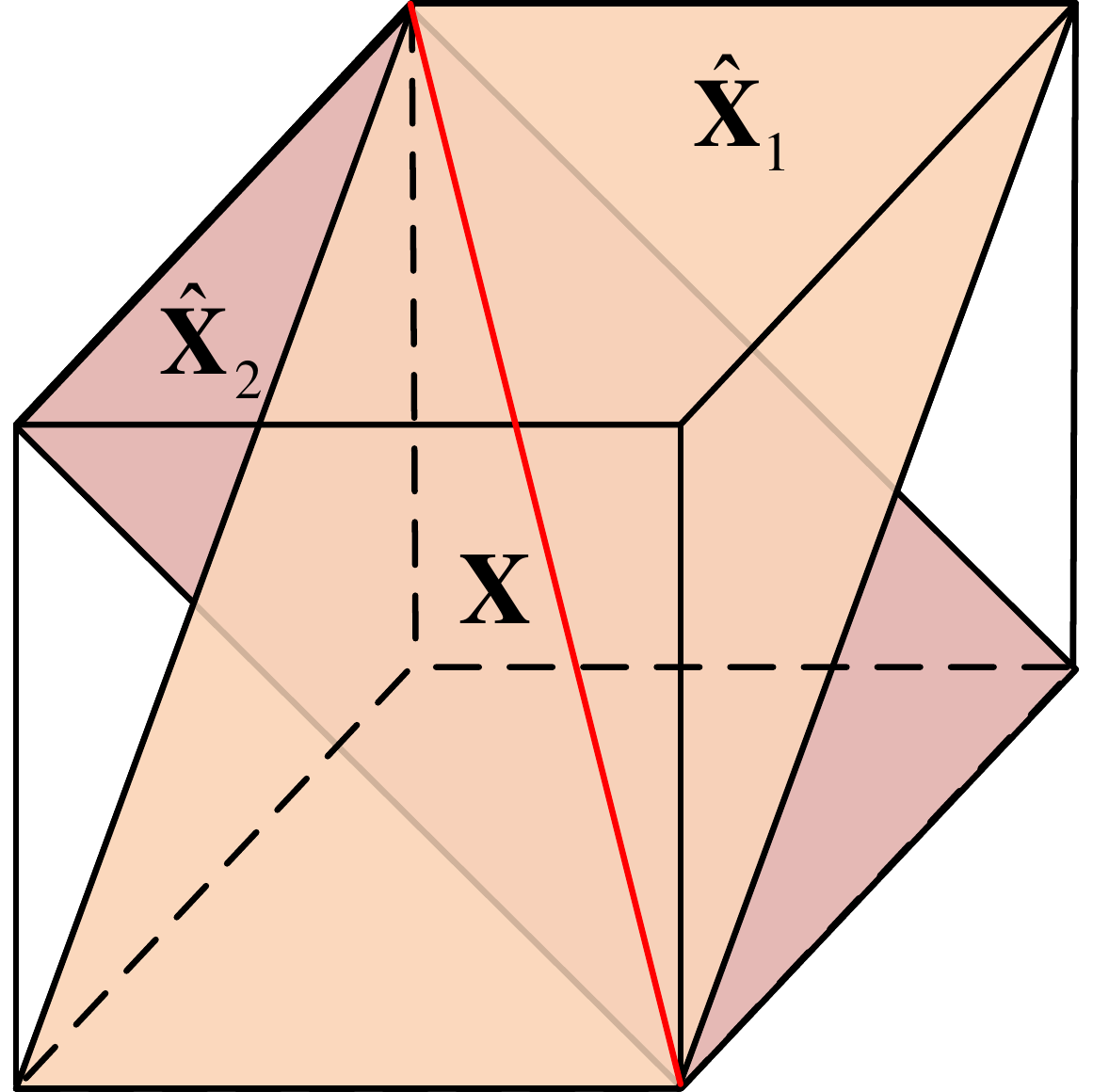}
  \caption{Illustration of constraining the data solution space of different ISAR forward model, where the red line represents the constrained ISAR image $\mathbf{X}$.}\label{fig:range space}
\end{figure}

\begin{figure*}[t]
  \centering
  \includegraphics[width =17.3cm, height =7cm]{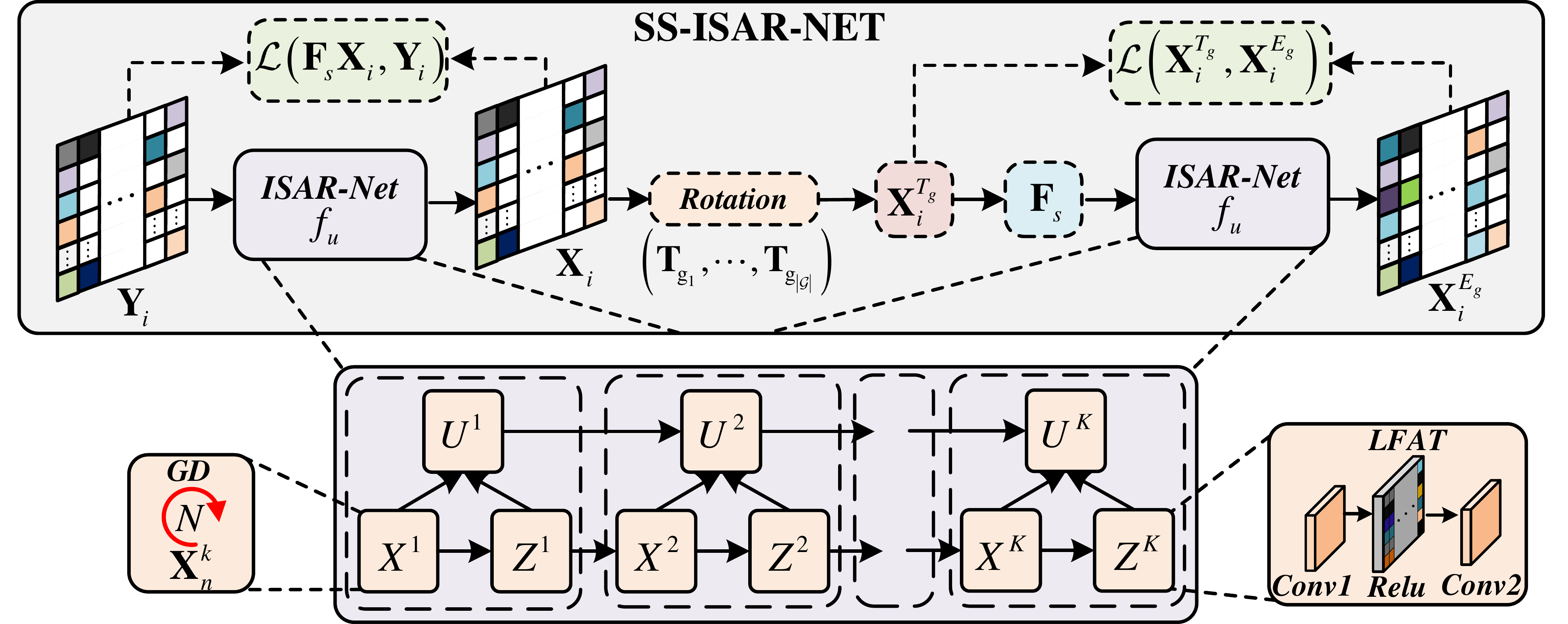}
  \caption{Illustration of the proposed SS-ISAR-Net, where the upper part is the overall architecture of SS-ISAR-Net with the equivariant constraint, and the lower part is the specific structure of the ISAR-Net with the LFAT module.}\label{fig:proposed}
\end{figure*}

The equivariant constraint is an inherent prior in imaging systems for natural images~\cite{chen2021equivariant}. It asserts that when an image transforms being input into the imaging system, the resulting output should exhibit the same transformation. Specifically, for an ISAR imaging system with input ISAR echo $\mathbf{Y}_s$ and corresponding forward model $\mathbf{F}_s$, denoted as $f\left( \mathbf{Y}_s, \mathbf{F}_s \right)$, the equivariant constraint can be formulated as
\begin{equation}
f\left( \mathbf{F}_s \mathbf{T}_g \mathbf{X}, \mathbf{F}_s \right) = \mathbf{T}_g f\left( \mathbf{F}_s \mathbf{X}, \mathbf{F}_s \right)
\label{equivariant}
\end{equation}
where $\mathbf{T}_g$ represents a transformation operator, such as image rotation, flipping, translation, etc. Based on~\eqref{subsolution}, the suboptimal solutions on both sides of~\eqref{equivariant} can be expressed as
\begin{equation}
\begin{gathered}
  {{{\mathbf{\hat X}}}_1} = f\left( {{{\mathbf{F}}_s}{{\mathbf{T}}_g}{\mathbf{X}},{{\mathbf{F}}_s}} \right) = {{\mathbf{T}}_g}{\mathbf{X}} + {{\mathbf{N}}_{{{\mathbf{F}}_s}}} \hfill \\[1mm]
  {{{\mathbf{\hat X}}}_2} ={{\mathbf{T}}_g}f\left( {{{\mathbf{F}}_s}{\mathbf{X}},{{\mathbf{F}}_s}} \right) = {{\mathbf{T}}_g}{\mathbf{X}} + {{\mathbf{T}}_g}{{\mathbf{N}}_{{{\mathbf{F}}_s}}} \hfill \\ 
\end{gathered} 
\end{equation}

Due to the equivariant constraint, i.e., $\mathbf{\hat X}_1 = \mathbf{\hat X}_2$, the following relationship can be derived
\begin{equation}
  {{\mathbf{T}}_g}{\mathbf{X}} + {{\mathbf{N}}_{{{\mathbf{F}}_s}}} = {{\mathbf{T}}_g}{\mathbf{X}} + {{\mathbf{T}}_g}{{\mathbf{N}}_{{{\mathbf{F}}_s}}}
\end{equation}

To ensure the validity of the above equation, the following condition must be satisfied
\begin{equation}
  {{\mathbf{N}}_{{{\mathbf{F}}_s}}} = {{\mathbf{T}}_g}{{\mathbf{N}}_{{{\mathbf{F}}_s}}} = {{\mathbf{N}}_{{{\mathbf{T}}_g}{{\mathbf{F}}_s}}}
  \label{eq:0}
\end{equation}

We know that~\eqref{eq:0} holds only when $\mathbf{T}_g{{\mathbf{F}}_s}$ is independent of ${{\mathbf{F}}_s}$. It occurs if and only if the null space of ${{\mathbf{F}}_s}$, denoted as $\mathcal{N}\left( {{{\mathbf{F}}_s}} \right)$, is reduced to $\emptyset$. This confirms that the equivariant constraint effectively eliminates the null space. In other words, the equivariant constraint extends the forward operator by incorporating transformations, forming a new operator ${{{\mathbf{T}}_g}{{\mathbf{F}}_s}}$. As shown in Fig.~\ref{fig:range space}, a simple schematic diagram demonstrates the reduction of the null space after introducing an equivariant constraint. 

Furthermore, due to the potentially highly sparse radar sampling, which results in a large null space, relying on a single transformation is typically insufficient. Instead, multiple image transformations are necessary to further constrain the solution space and improve ISAR reconstruction quality. 
By defining a set of transformations as $\mathcal{G}\left( {{{\mathbf{g}}_1},{{\mathbf{g}}_2}, \cdots ,{{\mathbf{g}}_{\left| \mathcal{G} \right|}}} \right)$ (where $|\mathcal{G}|$ is the total number of these transformations), an expanded forward operator ${{\mathbf{F}}_g}$ can be constructed, which integrates multiple equivariant constraints to refine the reconstruction process
\begin{equation}
{{\mathbf{F}}_g} = {\left[ {{{\mathbf{T}}_{{g_1}}}{{\mathbf{F}}_s},{{\mathbf{T}}_{{g_2}}}{{\mathbf{F}}_s}, \cdots ,{{\mathbf{T}}_{g_{\left| \mathcal{G} \right|}}}{{\mathbf{F}}_s}} \right]^T} \in {\mathbb{C}^{\left| \mathcal{G} \right|M \times N}}
\end{equation}
The necessary condition for eliminating the null space and obtaining the optimal ISAR image is the following
\begin{equation} \bigcap\limits_{{{\bf{g}}_i} \in {\cal G}} \mathcal{N}\left( {{{\mathbf{T}}_g}{{\mathbf{F}}_s}} \right) = \emptyset 
\label{condition}
\end{equation}

Condition \eqref{condition} ensures that the null space is minimized by leveraging multiple equivariant transformations, leading to the following necessary condition. 

\begin{theorem}\label{theorem1}
  A necessary condition for obtaining the optimal ISAR image through equivariant constraint is that the expanded forward operator satisfies $rank\left( {{{\mathbf{F}}_g}} \right) =N$.
\end{theorem}
\begin{proof}\label{proof11}
If the rank of the expanded forward operator is less than $N$, it indicates that the null space still exists, i.e., there is no guarantee that ${\mathbf{N}}\left( {{{\mathbf{F}}_s}} \right) = \emptyset$. This will lead to suboptimal solutions, ${\mathbf{\hat X}} = {\mathbf{X}} +{{\mathbf{N}}_{{{\mathbf{F}}_s}}} \ne {\mathbf{X}}$.
\end{proof}  

With the theoretical foundation of ISAR imaging using equivariant constraints established, the following will focus on the practical implementation of this approach.
Specifically, according to Theorem~\ref{theorem1}, 
the selection of image transformations will be guided by the condition that the extended ISAR forward operator ${{\bf{F}}_g}$ leads to an overdetermined system, i.e. $\left| \mathcal{G} \right| \times \gamma  > 1$ must be satisfied, ensuring the effectiveness of the equivariant constraint in reducing the grating lobes. Since ISAR images are typically centered within the image frame, this paper adopts rotation as the selected image transformation to enforce the equivariant constraint. 

The flowchart of the proposed SS-ISAR-Net is illustrated in Fig.~\ref{fig:proposed}, where the ISAR-Net ${f_u}\left(  \cdot  \right)$ is designed as described in Section~\ref{Sec:3-1}. The training loss function consists of two components: measurement consistency (MC) loss and equivariant consistency (EC) loss. Mathematically, it can be expressed as  
\begin{align}\label{eq:loss1}
&\mathcal{L} = \mathcal{L}_\mathrm{MC} + \alpha \mathcal{L}_\mathrm{EC} = \sum\limits_{i = 1}^I {\left[ {\left\| {{{\bf{Y}}_i} - {{\bf{F}}_s}{f_u}\left( {{{\bf{Y}}_i},{{\bf{F}}_s}} \right)} \right\|_F^2} \right.} \notag\\[-3mm]
 &+ \alpha \sum\limits_{g = 1}^{\left| {\cal G} \right|} {\left. {\left\| {{{\bf{T}}_g}{f_u}\left( {{{\bf{Y}}_i},{{\bf{F}}_s}} \right) - {f_u}\left( {{{\bf{F}}_s}{{\bf{T}}_g}{f_u}\left( {{{\bf{Y}}_i},{{\bf{F}}_s}} \right),{{\bf{F}}_s}} \right)} \right\|_F^2} \right]} 
\end{align}
where $\alpha$ is the regularization coefficient that balances the MC loss and EC loss, and $I$ is the number of incomplete ISAR echo data used for training. It should be noted that ${{\mathbf{F}}_s}{f_u}\left( {{{\mathbf{Y}}_i},{{\mathbf{F}}_s}} \right) = {{\mathbf{A}}_s}{f_u}\left( {{{\mathbf{Y}}_i},{{\mathbf{F}}_s}} \right){{\mathbf{B}}_s}$. With this formulation, the self-supervised learning of ISAR-Net based on the equivariant constraint, named SS-ISAR-Net, is completed. However, since the equivariant constraint does not account for noise, SS-ISAR-Net can work well in high signal-to-noise Ratio (SNR) scenarios. To extend operational robustness,  a self-denoising module to enhance robustness against noise interference is suggested below.

\begin{figure}[t]
  \centering
  \includegraphics[width =8cm, height =8.9cm]{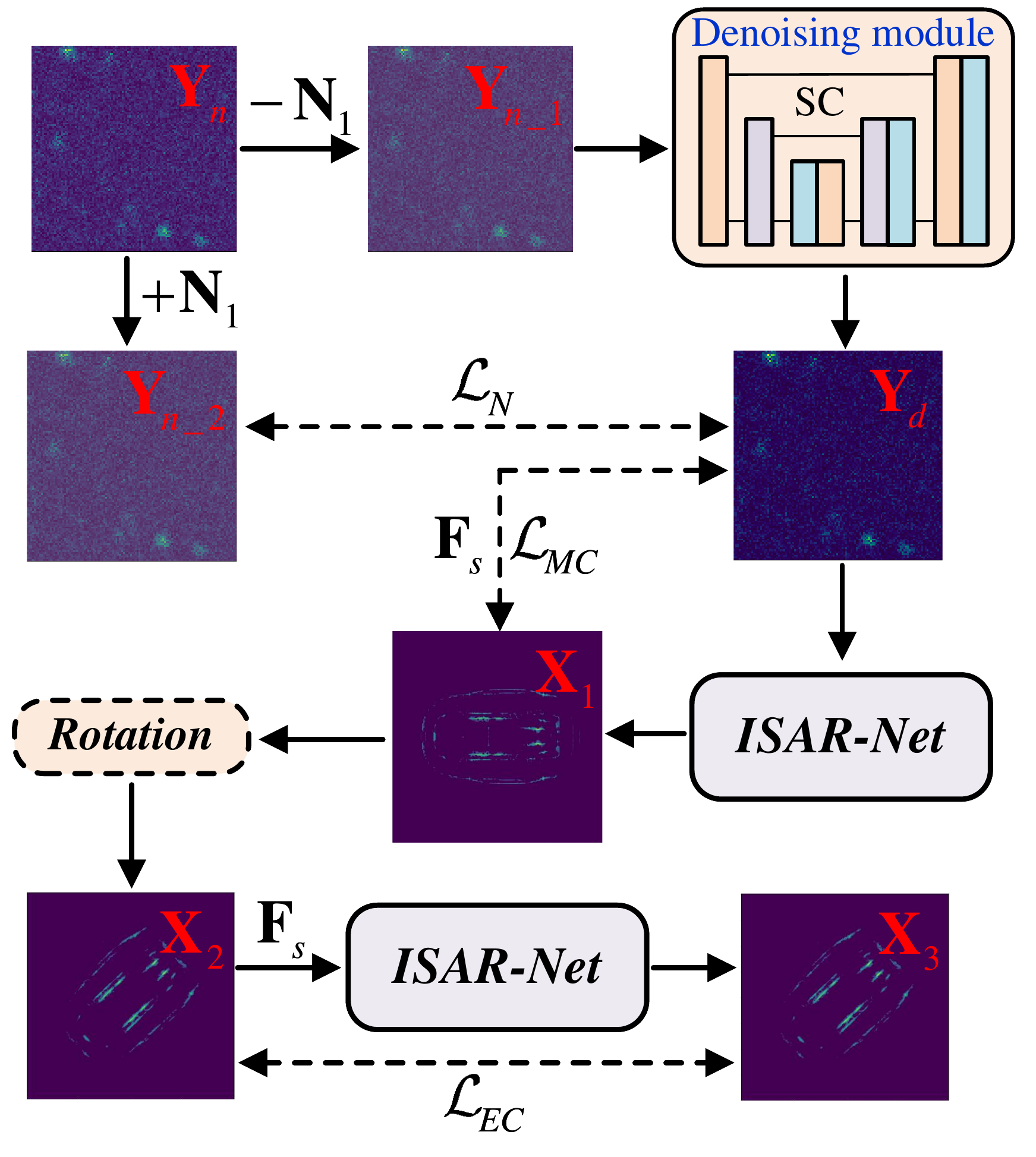}
  \caption{Illustration of Extending self-denoising to SS-ISAR-Net.}\label{fig:self-denoising}
\end{figure}

\subsection{Integrating self-denoising module to SS-ISAR-Net}\label{Sec:3-3}

The inherent non-equivariance property of noise in the proposed self-supervised method based on the equivariant constraint limits its robustness, particularly under low-SNR conditions. 
To address this issue, a self-denoising strategy is introduced to the SS-ISAR-Net framework, significantly improving its robustness by mitigating noise effects. 

Inspired by the concept of N2N~\cite{lehtinen2018noise2noise,qiao2024zero,pang2021recorrupted}, we adopt a similar strategy to denoise radar echo. Specifically, a radar echo corrupted by complex Gaussian noise can be expressed as
\begin{equation}
  {{\mathbf{Y}}_n} = {\mathbf{Y}} + {\mathbf{N}}
\end{equation}
where ${\mathbf{Y}}$ is the desired clean echo, and ${\mathbf{N}} \sim \mathcal{CN}\left( {0,{\sigma ^2}I} \right)$ denotes the complex-valued Gaussian noise with variance ${\sigma ^2}$. By assuming that the noise variance ${\sigma ^2}$ is known, two new re-corrupted echoes can be generated by adding and subtracting noise with the same variance to the original echo ${{\mathbf{Y}}_n}$ 
\begin{equation}
\begin{gathered}
  {{\mathbf{Y}}_{n\_1}} = {{\mathbf{Y}}_n} + {{\mathbf{N}}_1} = {\mathbf{Y}} + {\mathbf{N}} + {{\mathbf{N}}_1} = {\mathbf{Y}} + {{\mathbf{N}}_{n1}} \hfill \\
  {{\mathbf{Y}}_{n\_2}} = {{\mathbf{Y}}_n} - {{\mathbf{N}}_1} = {\mathbf{Y}} + {\mathbf{N}} - {{\mathbf{N}}_1} = {\mathbf{Y}} + {{\mathbf{N}}_{n2}} \hfill \\ 
\end{gathered} 
\label{Noise}    
\end{equation}
where ${\mathbf{N}_1} \sim \mathcal{CN}\left( {0,{\sigma ^2}I} \right)$ represents the added and subtracted complex-valued Gaussian noise. As formally proven in Appendix~\ref{proof1}, the resulting noise components ${{\mathbf{N}}_{n1}}$ and ${{\mathbf{N}}_{n2}}$ are statistically independent. Based on this property, for a denoising network ${f_d}\left(  \cdot  \right)$, using ${{\mathbf{Y}}_{n1}}$ as input and computing the mean squared error (MSE) loss with ${{\mathbf{Y}}_{n2}}$ is mathematically equivalent to computing the MSE loss with the clean echo ${{\mathbf{Y}}}$, which can be expressed as 
\begin{align} \label{eq:derivation1}
  &\mathbb{E}_{\mathbf{Y}_{n\_1}, \mathbf{Y}_{n\_2}} 
  \left\{ \left\| f_d \left( \mathbf{Y}_{n\_1} \right) - \mathbf{Y}_{n\_2} \right\|_F^2 \right\} \notag \\ 
  &= \mathbb{E}_{\mathbf{N}_{n1}, \mathbf{N}_{n2}} 
  \left\{ \left\| f_d \left( \mathbf{Y} + \mathbf{N}_{n1} \right) - \mathbf{Y} \right\|_F^2 \right\}
\end{align}
This property allows the network to learn noise removal without requiring access to clean data, making it an effective self-denoising approach.
Following a similar derivation, the MC loss for the denoised echo input to the imaging network can also be demonstrated to be equivalent to the loss computed with the clean echo $\mathbf{Y}$.
\begin{align} \label{eq:derivation2}
  &\mathbb{E}_{\mathbf{Y}_{n\_1}, \mathbf{Y}_{n\_2}} 
  \left\{ \left\| \mathbf{A} f_u \left( f_d \left( \mathbf{Y}_{n\_1} \right) \right) - \mathbf{Y}_{n\_2} \right\|_F^2 \right\} \notag \\ 
  &= \mathbb{E}_{\mathbf{N}_{n1}, \mathbf{N}_{n2}} 
  \left\{ \left\| \mathbf{A} \left( f_u \left( f_d \left( \mathbf{Y} + \mathbf{N}_{n1} \right) \right) - \mathbf{X} \right) \right\|_F^2 \right\}
\end{align}
This ensures that the self-denoising approach effectively mitigates noise interference while preserving measurement fidelity, further enhancing the robustness of the proposed SS-ISAR-Net.

The detailed derivations of~\eqref{eq:derivation1} and~\eqref{eq:derivation2} are provided in Appendix~\ref{proof2}. The flowchart illustrating the integration of self-denoising into SS-ISAR-Net is depicted in Fig.~\ref{fig:self-denoising}, where the denoising module is implemented by a simple U-Net \cite{ronneberger2015u}. Combining~\eqref{eq:loss1}, ~\eqref{eq:derivation1}, and~\eqref{eq:derivation2}, the final training loss function can be expressed as
\begin{align}\label{eq:loss2}
{\cal L} = {{\cal L}_N} + {\cal L}_{\text{MC}} + \alpha {{\cal L}_\text{EC}} = \sum\limits_{i = 1}^I {\left[ {\left\| {{{\bf{Y}}_d} - {{\bf{Y}}_{n\_2}}} \right\|_F^2} \right.} \notag\\
+ \left\| {{{\bf{F}}_s}{{\bf{X}}_1} - {{\bf{Y}}_{n\_2}}} \right\|_F^2 + \alpha \sum\limits_{g = 1}^{\left| {\cal G} \right|} {\left. {\left\| {{{\bf{X}}_2} - {{\bf{X}}_3}} \right\|_F^2} \right]} 
\end{align}
where
\begin{equation}
\begin{gathered}\notag
  {{\mathbf{Y}}_d} = {f_d}\left( {{{\mathbf{Y}}_{n\_1}}} \right),{{\mathbf{X}}_1} = {f_u}\left( {{{\mathbf{Y}}_d},{{\mathbf{F}}_s}} \right) \hfill \\[1mm]
  {{\mathbf{X}}_2} = {{\mathbf{T}}_g}{{\mathbf{X}}_1},{{\mathbf{X}}_3} = {f_u}\left( {{{\mathbf{F}}_s}{{\mathbf{X}}_2},{{\mathbf{F}}_s}} \right) \hfill 
\end{gathered} 
\end{equation}

So far, the self-denoising module has been introduced into SS-ISAR-Net to mitigate the impact of noise on the equivariant constraint, especially in the case of strong noise.

\section{Experiments}\label{Sec:4}
In this section, we first introduce several evaluation metrics used to characterize the performance of ISAR imaging methods. 
Then, both numerical simulations and experiments were conducted to assess the effectiveness of the proposed SS-ISAR-Net and compare it with a few state-of-the-art methods, i.e., RD method, Linear ADMM (LADMM)~\cite{wang20223}, Reinforcement Learning-based SwinRL-ADMM (supervised, requiring paired data)~\cite{li2025tuning}, SFH-ADMM-Net (supervised, requiring paired data)~\cite{wang2022efficient}, and Sup-ISAR-Net (supervised, requiring paired data), which refers to the proposed ISAR-Net in Section~\ref{Sec:3-1} trained using paired data.

\begin{figure*}[t]
  \centering
  \subfloat[]{\includegraphics[width =5.7cm, height =4.3cm]{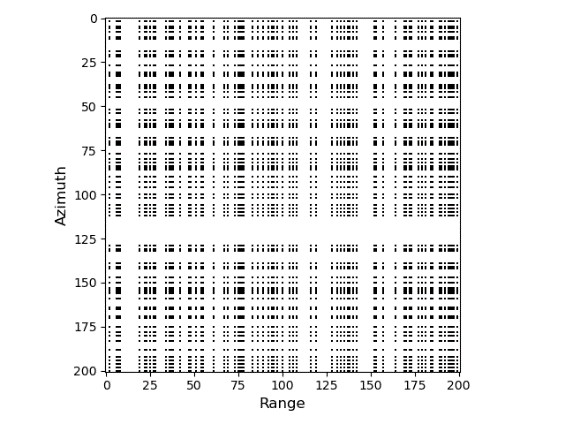}\label{fig:1}}
  \subfloat[]{\includegraphics[width =5.7cm, height =4.3cm]{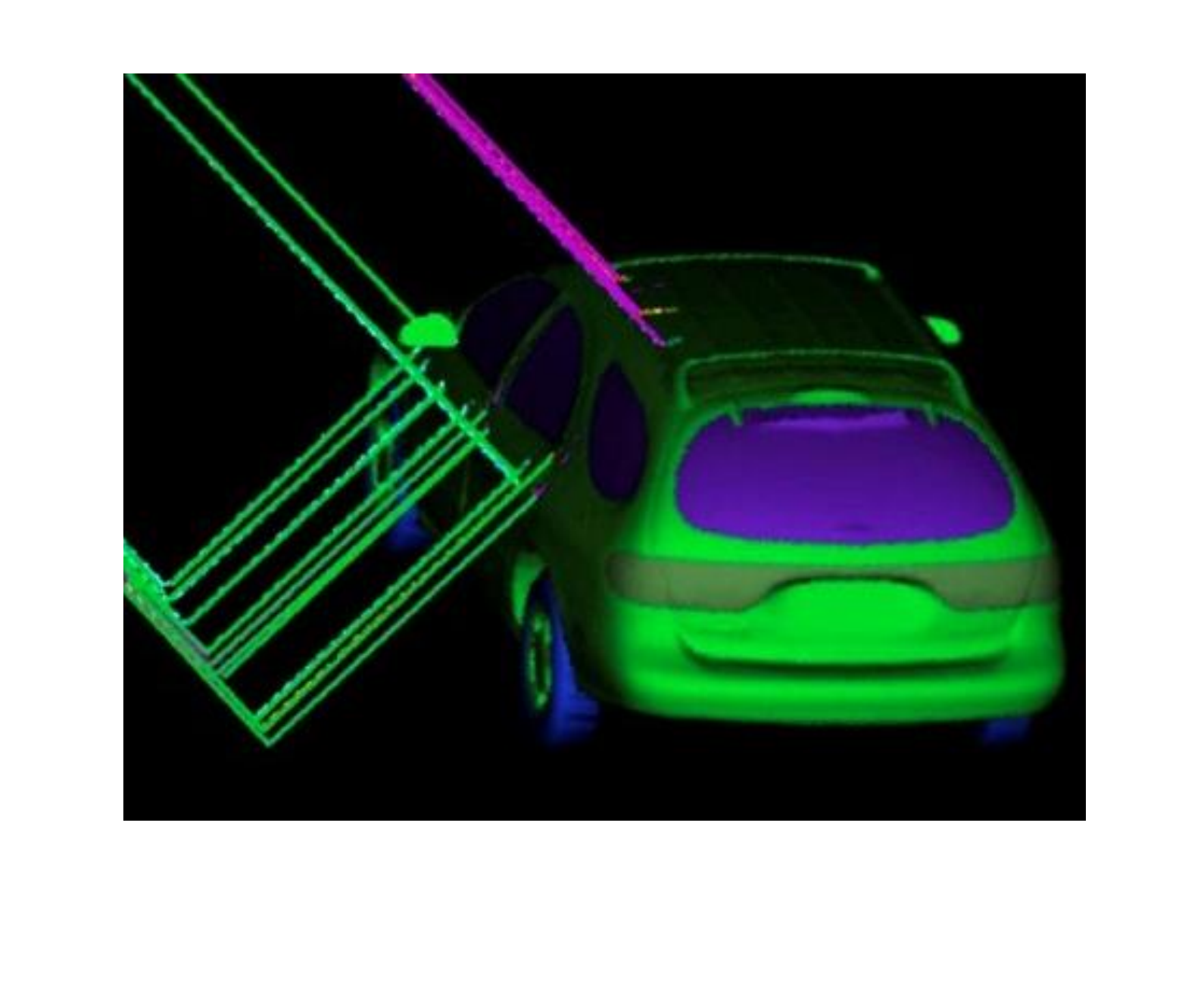}\label{fig:2}}
  \subfloat[]{\includegraphics[width =5.7cm, height =4.3cm]{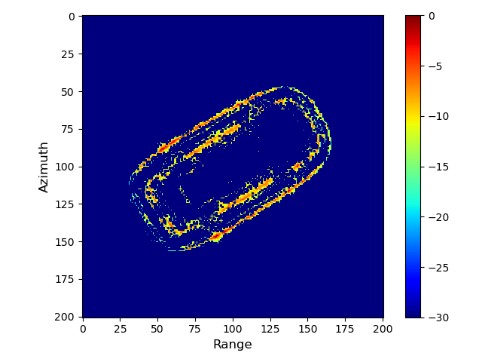}\label{fig:3}}
  \vspace{0.2cm}
  \caption{\protect\subref{fig:1} Sparse sampling mode of ISAR echo.
  \protect\subref{fig:2} Electromagnetic vehicle model.
  \protect\subref{fig:3} ISAR image of complete echo.}\label{fig:samplingmode}
\end{figure*}
\subsection{Description of Metrics}
To quantitatively assess the reconstruction performance of different imaging methods, this paper employs three key metrics: Normalized Mean Square Error (NMSE), Peak Signal-to-Noise Ratio (PSNR), and Structural Similarity Index Measure (SSIM). They are defined as follows
\begin{align}
\begin{gathered}
  {\text{NMSE}}\left( {{\mathbf{X}},{\mathbf{\hat X}}} \right) = \frac{{\left\| {{\mathbf{X}} - {\mathbf{\hat X}}} \right\|_F^2}}{{\left\| {\mathbf{X}} \right\|_F^2}}\notag \\ 
  {\text{PSNR}}\left( {{\mathbf{X}},{\mathbf{\hat X}}} \right) = 10{\log _{10}}\frac{{{{\left( {{{\text{p}}_{{\text{max}}}}} \right)}^2}}}{{\frac{1}{{PQ}}\left\| {{\mathbf{X}} - {\mathbf{\hat X}}} \right\|_F^2}}\notag \\ 
  {\text{SSIM}}\left( {{\mathbf{X}},{\mathbf{\hat X}}} \right) = \frac{{\left( {2{\mu _{\mathbf{X}}}{\mu _{{\mathbf{\hat X}}}} + {c_1}} \right)\left( {2{\sigma _{{\mathbf{X\hat X}}}} + {c_2}} \right)}}{{\left( {\mu _{\mathbf{X}}^2 + \mu _{{\mathbf{\hat X}}}^2 + {c_1}} \right)\left( {\sigma _{\mathbf{X}}^2 + \sigma _{{\mathbf{\hat X}}}^2 + {c_2}} \right)}} \notag\\ 
\end{gathered} 
\end{align}
where $\mathbf{X}$ and $\mathbf{\hat X}$ are the ground truth and the reconstructed image, respectively, and the ISAR image of complete echo is used as the ground truth.  
$\rm{p}_\text{max}$ indicates the maximum amplitude in the image. In this paper, due to normalization, $\rm{p}_\text{max}=1$.  ${\mu _{\mathbf{X}}},\, {\mu _{{\mathbf{\hat X}}}},\, {\sigma _{\mathbf{X}}},\, {\sigma _{{\mathbf{\hat X}}}}$ denote the means and standard deviations of ${\mathbf{X}}$ and ${\mathbf{\hat X}}$, while ${\sigma _{{\mathbf{X\hat X}}}}$ is the corresponding covariance matrix. $c_1 = ( k_1 \rm{p}_\text{max} )^2,\, c_2 = ( k_2 \rm{p}_\text{max} )^2$, and $k_1 = 0.01,\, k_2 = 0.03$. A lower NMSE indicates a smaller reconstruction error, whereas higher PSNR and SSIM values suggest that the reconstruction result is more faithful to the ground truth.

\subsection{Description of networks training and implementation details}
These networks are trained using the GOTCHA civilian vehicle electromagnetic simulation dataset, which comprises nine vehicle classes~\cite{dungan2010civilian}. Among these, eight classes are designated for training, while one class is reserved for testing. Reference ground truth images are generated through back projection (BP) imaging of complete echo data across multiple pitch angles. This procedure results in 32 images, each with size $200\times 200$. After amplitude normalization, targets with amplitudes below 0.01 are discarded. To augment the training dataset, image rotation, and flipping are applied for data augmentation, increasing the total number of training ground truth images to 192. These images are subsequently used to construct incomplete ISAR echo training data pairs based on~\eqref{fig:ISAR}. Note that the range and azimuth down-sampling rates for the training are set at approximately 64$\,\%$, resulting in a total sampling rate of around 40$\,\%$. Additionally, to independently validate methods described in Secs~\ref{Sec:3-2} and~\ref{Sec:3-3} in different noise conditions, the training process is divided into weak-noise and strong-noise categories. These paired training data are utilized for training SFH-ADMM-Net, SwinRL-ADMM, and Sup-ISAR-Net. In contrast, the proposed SS-ISAR-Net is trained solely on incomplete ISAR echo data and has never been exposed to ground truth images throughout the learning process.

The equivariant constraint is implemented through image rotation transformations ${{\mathbf{T}}_{\text{g}}}$ during network training. Specifically, each training iteration randomly selects three rotation angles uniformly distributed between $0^\circ$ to $360^\circ$. As proved in Theorem~\ref{theorem1}, this configuration satisfies the mathematical condition $3\times40\,\%>1$ required for effective constraint enforcement. The sampling diagram, the electromagnetic vehicle model, and an ISAR image of complete echo are illustrated in Fig.~\ref{fig:samplingmode}. 

In this paper, the complete SS-ISAR-Net architecture consists of $K=12$ iterative modules (with justification for this parameter choice provided in Section~\ref{Sec:5-1}). Both convolutional layers $\mathcal{C}_1$ and $\mathcal{C}_2$ in the LFAT module use a kernel size of seven, and the impact of different kernel sizes on this module will be analyzed in Sec~\ref{Sec:5-2}. 
The EC regularization coefficients $\alpha$ in~\eqref{eq:loss1} and~\eqref{eq:loss2} are set as 1, which will be analyzed in Sec~\ref{Sec:5-3}. 
For comparison methods, since LADMM requires manually selected hyperparameters, it will change with the target type and signal SNR, and we will manually adjust it to achieve a good result. The parameter settings of SwinRL-ADMM are consistent with the original paper in~\cite{li2025tuning}. SFH-ADMM-Net has 10 iterative modules as described in~\cite{wang2022efficient}. Sup-ISAR-Net has the same parameters as SS-ISAR-Net, but the training strategy is different.
In addition, all the networks were trained for 200 epochs using the Adam optimizer. The initial learning rate was set to $1 \times 10^{-4}$ and halved every 50 epochs. All the methods were implemented on an NVIDIA A100 GPU using the PyTorch framework.

\begin{figure*}[t]
  \hspace{-0.3cm}
  \includegraphics[width =18cm, height =11.3cm]{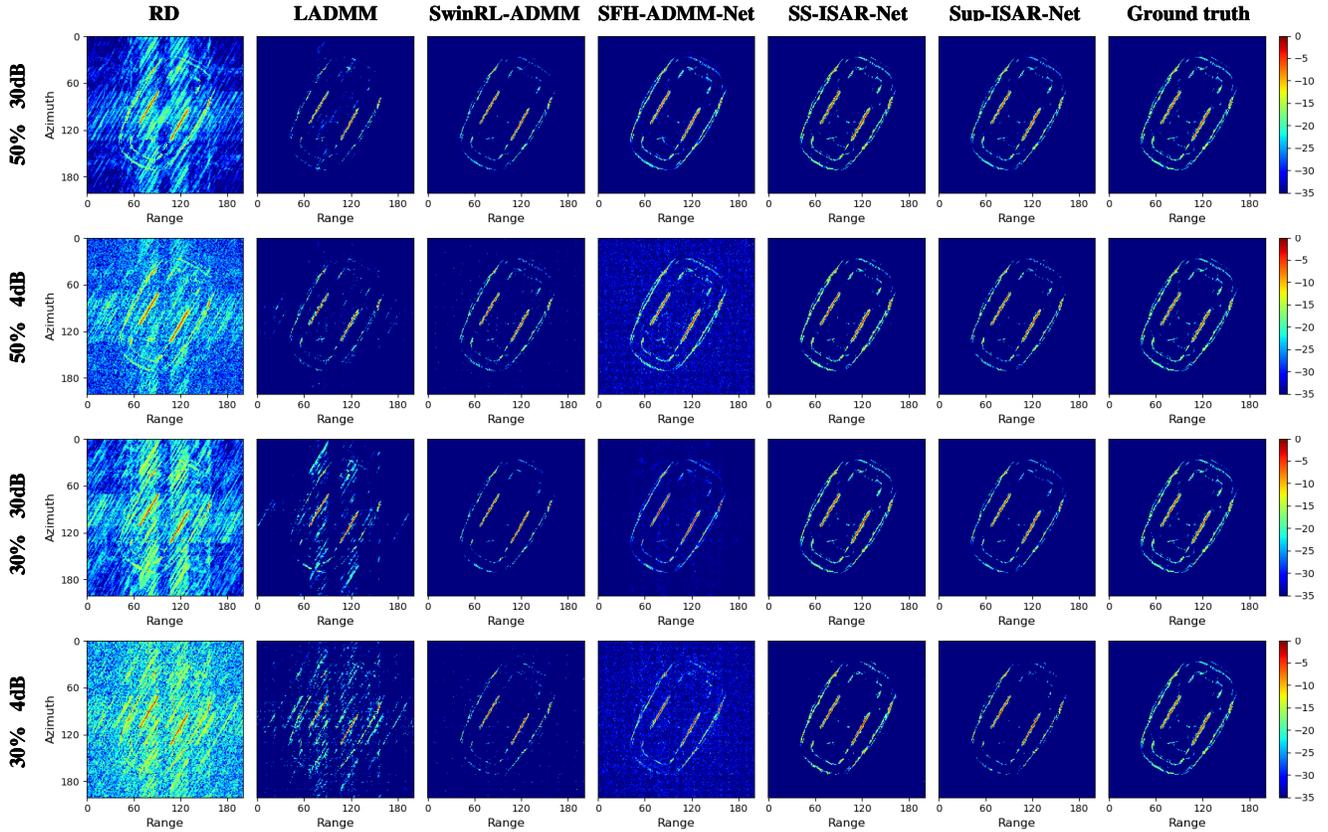}
  \caption{ISAR imaging of electromagnetic vehicle model with different sampling rate and noise level. From left to right, they are RD, LADMM, SwinRL-ADMM, SFH-ADMM-Net, the proposed SS-ISAR-Net, Sup-ISAR-Net and corresponding GT.}
\label{fig:result_1}
\end{figure*}

\subsection{Civilian vehicle electromagnetic testing}
\label{Sec:4-2}
In this section, we evaluate the performance of the proposed SS-ISAR-Net using the reserved test class from the GOTCHA dataset. To rigorously assess its performance, we conduct experiments under varying sampling and noise conditions: two different down-sampling ratios $\gamma=50\,\%$ and $\gamma=30\,\%$ and different noise levels were used: noise was added to the echo SNR to about $30\,\mathrm{dB}$ and $4\,\mathrm{dB}$ respectively. This was done to make the parameters inconsistent with the training dataset and further verify the generalization of the data sampling ratios. 

The reconstructed ISAR images obtained using different methods under various down-sampling rates and SNR are presented (Fig.~\ref{fig:result_1}). From top to bottom, the rows correspond to the ISAR reconstruction results under the following conditions: a 50$\,\%$ down-sampling rate with SNR levels of $30\,\mathrm{dB}$ and $4\,\mathrm{dB}$, followed by a 30$\,\%$ down-sampling rate with SNR levels of $30\,\mathrm{dB}$ and $4\,\mathrm{dB}$.

Fig.~\ref{fig:result_1} shows that the RD algorithm produces severe grating lobes, significantly degrading image quality. The LADMM, which relies on approximation techniques, fails under low sampling rates. While SwinRL-ADMM achieves effective reconstruction through reinforcement learning-based hyperparameter optimization, it does not recover fine structural details very well. SFH-ADMM-Net shows a noticeable performance decline under low sampling rates and high noise levels, indicating poor generalization capability.
By contrast, our proposed SS-ISAR-Net maintains robust performance across all operational conditions, matching the reconstruction quality of supervised Sup-ISAR-Net while demonstrating superior adaptability to extreme scenarios. 

The quantitative evaluation metrics, presented in~Table~\ref{Tab:table1}, further validate this observation. Notably, SS-ISAR-Net demonstrates consistent performance advantages over Sup-ISAR-Net. This superiority is attributed to two key factors: (1) Sup-ISAR-Net's performance is constrained by limited training data availability, which adversely affects its learned end-to-end mapping; (2) SS-ISAR-Net circumvents this limitation by employing equivariant constraints to extend the forward operator, thereby learning a more robust solution strategy rather than direct end-to-end mapping. As evidenced in Fig.~\ref{fig:loss_sup_zs}, while Sup-ISAR-Net achieves higher PSNR during training due to its supervised learning paradigm, SS-ISAR-Net exhibits superior testing performance. This performance difference demonstrates the enhanced generalization capability of our self-supervised learning approach.

\begin{figure}[t]
    \hspace{-0.3cm}
    \includegraphics[width =8.7cm, height =4.5cm]{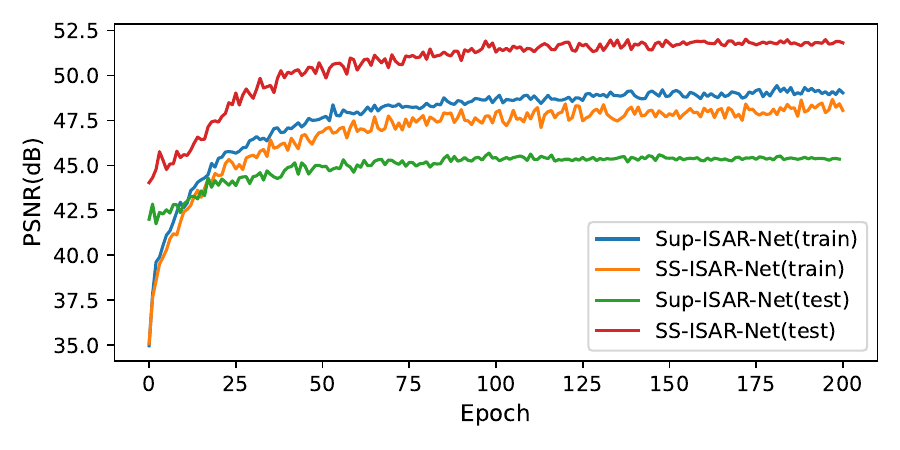}
    \vspace{-0.7cm}
    \caption{Reconstruction PSNRs of training epoch for the SS-ISAR-Net and Sup-ISAR-Net for training and testing.}
    \label{fig:loss_sup_zs}
\end{figure}

\begin{table}[!htbp]
  \centering
  \caption{Quantitative performance of civilian vehicle electromagnetic simulation with different down-sampling rate and SNR. \textcolor{red}{Red} and \textcolor{blue}{Blue} denote the best and second best quality of each evaluation method, respectively.}
  \resizebox{8.5cm}{4.6cm}{
  \begin{tabular}{cccccc}
    \toprule
   $\gamma$          & SNR           & Methods      & NMSE ($\downarrow $)   & PSNR ($\uparrow $)   & SSIM ($\uparrow $)    \\ \midrule
  \multirow{12}{*}{50$\,\%$}&\multirow{6}{*}{$30\,\mathrm{dB}$}  & RD     & 1.0462 & 26.1355 & 0.1785 \\ 
      &    & LADMM        & 0.1237 & 35.4072 & 0.8882 \\ 
      &    & SwinRL-ADMM   & 0.0444 & 39.8624 & 0.9620 \\ 
      &  & SFH-ADMM-Net & 0.0644 & 38.2442 & 0.8797  \\ 
      &    & Sup-ISAR-Net & \textcolor{blue}{0.0089}& \textcolor{blue}{46.8235}& \textcolor{blue}{0.9912}\\
      &    & SS-ISAR-Net  & \textcolor{red}{0.0023} &\textcolor{red}{52.7479}&\textcolor{red}{0.9977}\\ 
 \cline{2-6} \rule{0pt}{12pt}
      & \multirow{6}{*}{$4\,\mathrm{dB}$} & RD   &1.7694 & 23.8535 & 0.1168 \\ 
      &    & LADMM        & 0.1433 & 34.7686 & 0.8373  \\ 
      &    & SwinRL-ADMM   & 0.1025 & 36.2256 & 0.7879 \\ 
      &     & SFH-ADMM-Net &0.2775 & 31.8991 & 0.3089  \\
        &    & Sup-ISAR-Net &\textcolor{blue}{0.0553}&\textcolor{blue}{38.9054}&\textcolor{blue}{0.9591}\\
      &    & SS-ISAR-Net  &\textcolor{red}{0.0437}&\textcolor{red}{39.9283}&	\textcolor{red}{0.9661}\\ 
 \midrule
  \multirow{12}{*}{30$\,\%$} & \multirow{6}{*}{$30\,\mathrm{dB}$} & RD  & 2.3034 & 22.7081 & 0.1084             \\ 
      &    & LADMM        & 0.2839 & 31.7994 & 0.7303   \\ 
      &    & SwinRL-ADMM   & 0.0847 & 37.0550 & 0.9273\\ 
      &    & SFH-ADMM-Net & 0.2305 & 32.7059 & 0.6154\\
    &    & Sup-ISAR-Net &\textcolor{blue}{0.0388}&\textcolor{blue}{40.4431}&\textcolor{blue}{0.9617}\\
      &    & SS-ISAR-Net  &\textcolor{red}{0.0083}&\textcolor{red}{47.1412}&\textcolor{red}{0.9912}\\ 
       \cline{2-6} \rule{0pt}{12pt}
      & \multirow{6}{*}{$4\,\mathrm{dB}$}  & RD     & 3.9866 & 20.3258 & 0.0731                  \\ 
      &    & LADMM        &0.6118 & 28.4659 & 0.4494                \\
      &    & SwinRL-ADMM   &0.1686 & 34.0623 & 0.7940        \\
      &    & SFH-ADMM-Net &0.4246 & 30.0523 & 0.2878            \\ 
    &    & Sup-ISAR-Net &\textcolor{blue}{0.1196}&	\textcolor{blue}{35.5531}&\textcolor{blue}{0.9196}\\
      &    & SS-ISAR-Net  &\textcolor{red}{0.0825}&\textcolor{red}{37.1660}&\textcolor{red}{0.9405}\\ 
 \bottomrule
  \end{tabular}}
  \label{Tab:table1}
  \end{table}

\subsection{Experiments with Electromagnetic Satellite Model}\label{Sec:4-4}

\begin{figure}[t]
\centering
\vspace{-0.5cm}
\includegraphics[width =8.5cm, height =4.6cm]{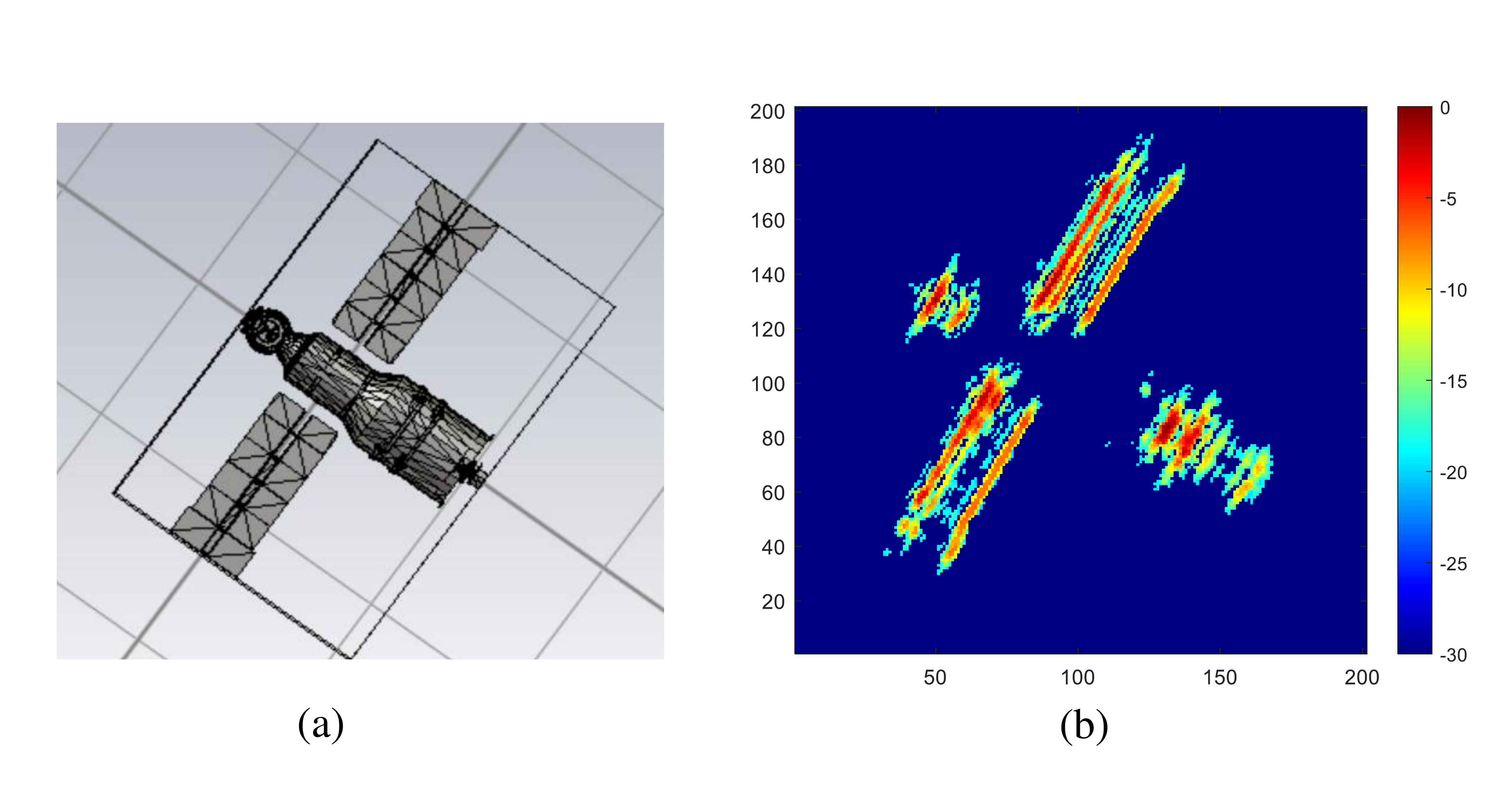}
\caption{(a) Electromagnetic satellite model.\,(b) ISAR image of its complete echo.}\label{fig:CST}
\end{figure}
  
\begin{figure*}[t]
\hspace{-0.3cm}
\includegraphics[width =18cm, height =6cm]{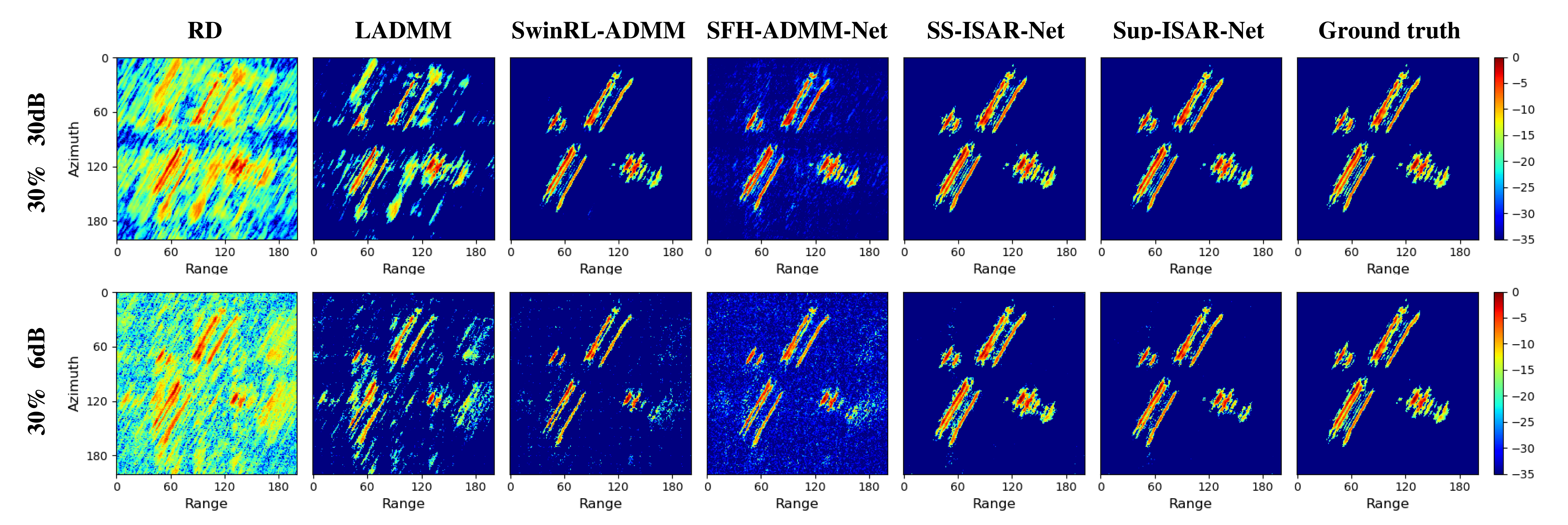}
\caption{ISAR imaging of electromagnetic satellite model with different noise level with 30$\,\%$ down-sampling rate. From left to right, they are RD, LADMM, SwinRL-ADMM, SFH-ADMM-Net, the proposed SS-ISAR-Net, Sup-ISAR-Net, and the corresponding ISAR image of complete echo.}\label{fig:result_2}
\end{figure*}

In this section, the effectiveness of the proposed SS-ISAR-Net is validated using electromagnetic simulation data from a satellite model in CST MICROWAVE STUDIO. The radar operates with a bandwidth of $4\,\mathrm{GHz}$, an azimuth synthesis angle of $4^\circ$, and a carrier frequency of $14\,\mathrm{GHz}$. The electromagnetic satellite model and its ISAR image obtained from complete echoes are illustrated in Fig.~\ref{fig:CST}. The complete echo is downsampled to 30$\,\%$, and complex Gaussian noise was added to synthesize two sets of radar echoes with the SNR of $30\,\mathrm{dB}$ and $6\,\mathrm{dB}$, respectively.

Fig.~\ref{fig:result_2} presents a comprehensive comparison of ISAR reconstruction performance across different methods. The RD algorithm suffers from prominent grating lobes caused by severe down-sampling artifacts. While LADMM incorporates iterative optimization techniques, it still fails to effectively suppress these artifacts. SwinRL-ADMM, by leveraging self-learning hyperparameters, successfully mitigates grating lobes but at the cost of losing finer structural details. Although SFH-ADMM-Net demonstrates superior performance compared to conventional approaches, the persistent artifacts in its reconstructions reveal fundamental limitations in its generalization capacity.

In marked contrast, both the proposed SS-ISAR-Net and supervised Sup-ISAR-Net demonstrate robust suppression of grating lobes in various noise conditions while maintaining exceptional structural fidelity. Notably, SS-ISAR-Net attains supervised-comparable generalization performance without requiring ISAR ground truth images during training. This breakthrough is made possible by two key factors: (1) a geometrically constrained equivariant prior that preserves target characteristics, and (2) the N2N self-supervised denoising framework that ensures reconstruction stability under noise perturbation.

Table~\ref{Tab:table2} provides a quantitative performance comparison between different methods. The proposed SS-ISAR-Net consistently achieves a lower NMSE and higher PSNR and SSIM, further demonstrating its effectiveness and robustness.

 \begin{table}[!htbp]
    \centering
    \caption{Quantitative performance of Electromagnetic Satellite Model with different SNR. \textcolor{red}{Red} and \textcolor{blue}{blue} denote the best and second best quality of each evaluation method, respectively.}
    \resizebox{8.5cm}{2.45cm}{
    \begin{tabular}{cccccc}
      \toprule
      $\gamma$          & SNR           & Methods      & NMSE ($\downarrow $)    & PSNR ($\uparrow $)   & SSIM ($\uparrow $)    \\ \midrule
    \multirow{12}{*}{30$\,\%$}&\multirow{6}{*}{$30\,\mathrm{dB}$}  & RD     & 1.8096 & 17.4470 & 0.1146\\ 
        &    & LADMM        & 0.3598 & 24.4617 & 0.6058 \\ 
        &    & SwinRL-ADMM   & 0.0280 & 35.5528 & 0.9730 \\ 
        &  & SFH-ADMM-Net & 0.1152 & 29.4070 & 0.4682 \\
     &    & Sup-ISAR-Net & \textcolor{blue}{0.0133}&\textcolor{blue}{38.7947}&\textcolor{blue}{0.9835} \\
        &    & SS-ISAR-Net  & \textcolor{red}{0.0045}&	\textcolor{red}{43.4626}&\textcolor{red}{0.9962}\\ 
 \cline{2-6} \rule{0pt}{12pt}
        & \multirow{6}{*}{$6\,\mathrm{dB}$} & RD    & 1.9226 & 17.1839 & 0.0927\\ 
        &    & LADMM        & 0.3997 & 24.0059 & 0.4130 \\ 
        &    & SwinRL-ADMM   & 0.2059 & 26.8862 & 0.5895 \\ 
        &  & SFH-ADMM-Net & 0.2515 & 26.0180 & 0.2386 \\ 
        &    & Sup-ISAR-Net & \textcolor{blue}{0.0740}&	\textcolor{blue}{31.3296}&	\textcolor{blue}{0.9356}         \\
        &    & SS-ISAR-Net  & \textcolor{red}{0.0545}&	\textcolor{red}{32.6555}&	\textcolor{red}{0.9531}\\  \bottomrule
    \end{tabular}}
    \label{Tab:table2}
\end{table}

\subsection{Experiments with Measurement Data}

\begin{figure*}[!htbp]
  \subfloat[]{\includegraphics[width =8.5cm, height =3.6cm]{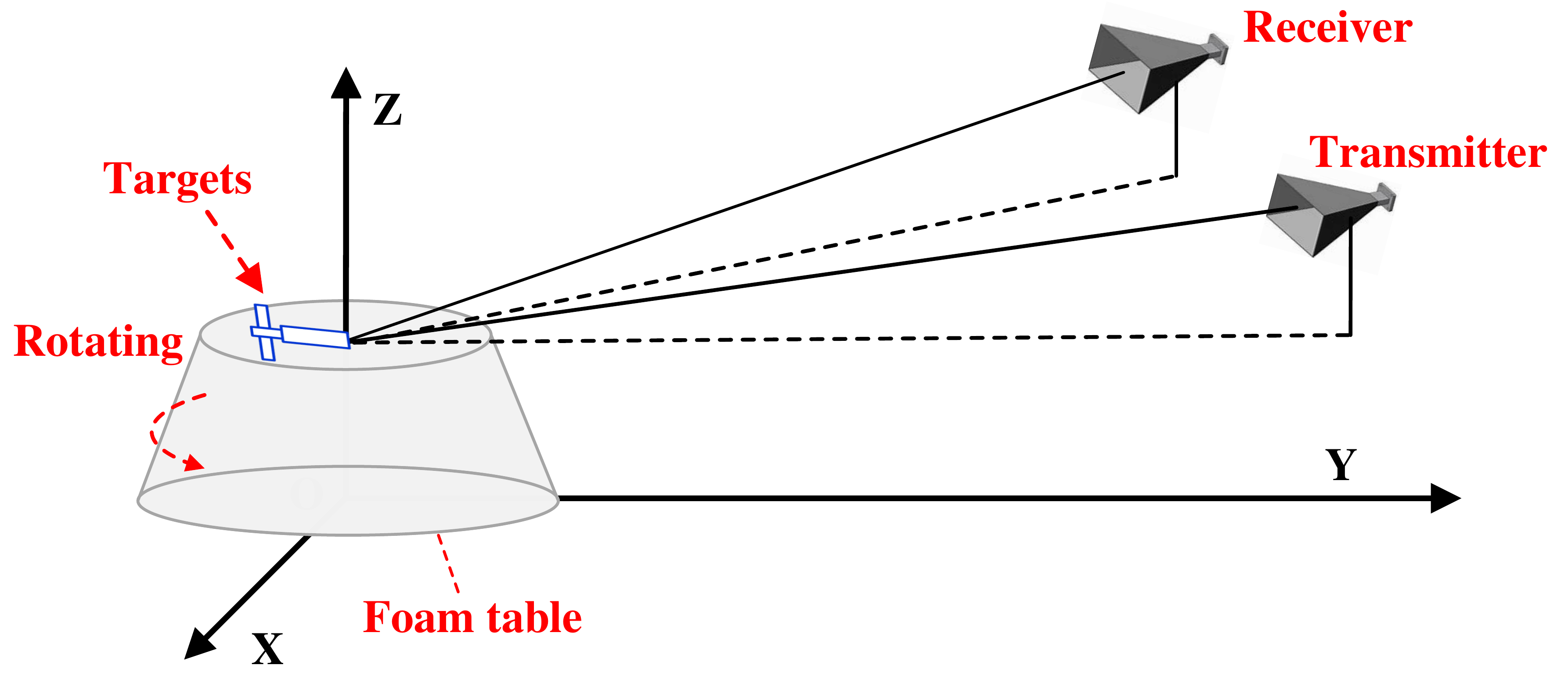}\label{fig:anshi1}}
  \hspace{0.3cm}
  \subfloat[]{\includegraphics[width =8.5cm, height =3.3cm]{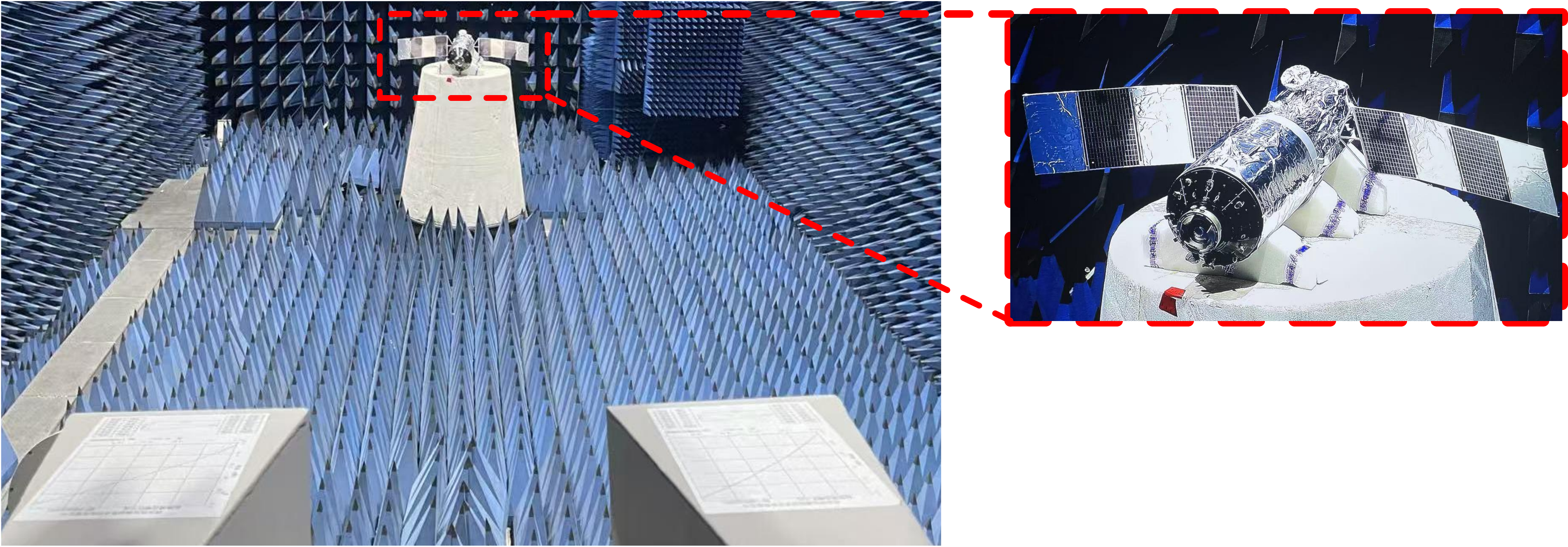}\label{fig:anshi2}}
  \caption{\protect\subref{fig:anshi1} Radar observation geometry of the anechoic chamber experiment.
  \protect\subref{fig:anshi2} Real-world images of the scenarios of the anechoic chamber experiment.}\label{fig:anshi}
\end{figure*}

  \begin{figure*}[!htbp]
    \centering
  \includegraphics[width =17.5cm, height =5.8cm]{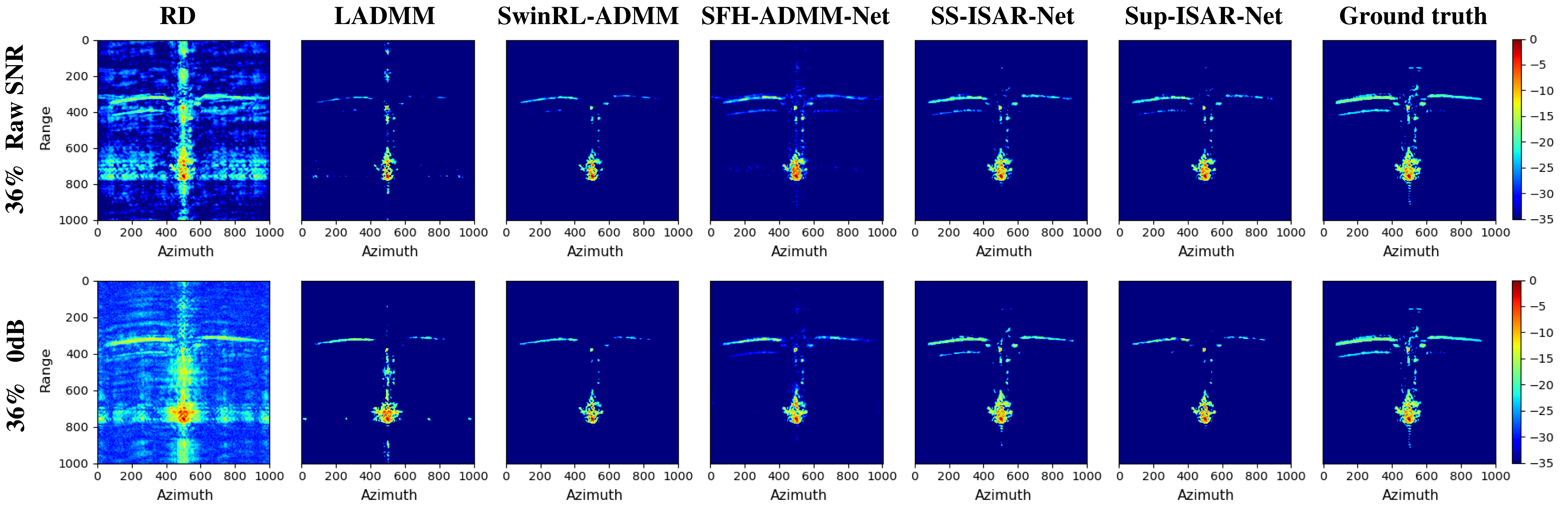}
    \caption{ISAR imaging of measurement data of a satellite model within a microwave anechoic chamber with different noise level with 36$\,\%$ down-sampling rate. From left to right, they are RD method, LADMM, SwinRL-ADMM, SFH-ADMM-Net, the proposed SS-ISAR-Net, Sup-ISAR-Net, and the corresponding ISAR image of complete echo.}
    \label{fig:result_3}
  \end{figure*} 

This section presents experimental validation of the proposed method using measurement data of a satellite model collected in a microwave anechoic chamber. The radar operates at a carrier frequency of $15\,\mathrm{GHz}$, with a bandwidth of $6\,\mathrm{GHz}$ and an azimuth synthesis angle of $20^\circ$. Figures~\ref{fig:anshi}\subref{fig:anshi1} and \ref{fig:anshi}\subref{fig:anshi2} respectively depict the experimental observation setup and corresponding target images. Due to the high SNR of the anechoic chamber echoes, additional noise is added to synthesize the echo with the SNR of $0\,\mathrm{dB}$ in order to evaluate the performance of the proposed method under noisy conditions. The complete echo data are downsampled to 36$\,\%$.

The imaging results for the echoes with the raw SNR and $0\,\mathrm{dB}$ are presented in Fig.~\ref{fig:result_3}. It can be seen that the RD method suffers the most significant degradation. LADMM partially suppresses these artifacts, but does not accurately reconstruct the geometric structure of the target. SwinRL-ADMM effectively reconstructs the basic target structure under both noise conditions, yet it loses finer details, such as partial solar panels. SFH-ADMM-Net exhibits poor performance under strong noise conditions, with some residual grating lobes still remaining. In contrast, SS-ISAR-Net achieves a more accurate overall reconstruction of the satellite model even under strong noise conditions. Table~\ref{Tab:table3} presents a quantitative performance comparison of different methods in the measurement data experiment. The proposed SS-ISAR-Net consistently produces superior ISAR images, with lower NMSE and higher PSNR and SSIM, which demonstrates its effectiveness and superiority, as it is trained solely on incomplete echoes, showcasing its robustness in real-world scenarios.

\begin{table}[t]
  \centering
  \caption{Quantitative performance of Measurement Data in microwave anechoic chamber with different SNR. \textcolor{red}{Red} and \textcolor{blue}{blue} denote the best and second best quality of each evaluation method, respectively.}
  \resizebox{8.5cm}{2.45cm}{
  \begin{tabular}{cccccc}
    \toprule
   $\gamma$          & SNR           & Methods      & NMSE ($\downarrow $)    & PSNR ($\uparrow $)   & SSIM ($\uparrow $)    \\ \midrule
  \multirow{12}{*}{36$\,\%$}&\multirow{6}{*}{Raw SNR}  & RD  & 1.4084 & 27.9675 & 0.2257 \\ 
      &    & LADMM        &0.2981 & 34.7114 & 0.9339  \\ 
      &    & SwinRL-ADMM   &0.2007 & 36.4302 & 0.9518 \\ 
      &  & SFH-ADMM-Net &0.2546 & 35.3966 & 0.6235  \\ 
      &    & Sup-ISAR-Net & \textcolor{blue}{0.1246}&\textcolor{blue}{38.4981}&\textcolor{blue}{0.9621} \\
      &    & SS-ISAR-Net &\textcolor{red}{0.0802}&\textcolor{red}{40.4119}&\textcolor{red}{0.9739}\\ 
     \cline{2-6} \rule{0pt}{12pt}
      & \multirow{6}{*}{$0\,\mathrm{dB}$} & RD  &  3.5506 & 23.9519 & 0.0602\\ 
      &    & LADMM        & 0.3257 & 34.3262 & 0.8867 \\ 
      &    & SwinRL-ADMM   &0.1884 & 36.7035 & 0.9288 \\ 
      &  & SFH-ADMM-Net &0.3953 & 33.4854 & 0.3469\\
      &    & Sup-ISAR-Net & \textcolor{blue}{0.1683}&	\textcolor{blue}{37.1952}&	\textcolor{blue}{0.9503}  \\ 
      &    & SS-ISAR-Net  &\textcolor{red}{0.0839}&\textcolor{red}{40.2161}&	\textcolor{red}{0.9677}\\ 
       \bottomrule
  \end{tabular}}\label{Tab:table3}
  \end{table}

\section{Discussion}
\label{Sec:5}
This section investigates the impact of key hyperparameters on network performance, including the number of iterations of the SS-ISAR-Net and the equivariant constraint (EC) regularization coefficient. Furthermore, an ablation study is conducted to assess the role of the LFAT module in the learning process.

\subsection{Iterative number $K$ of the SS-ISAR-Net}\label{Sec:5-1}
The number of iterative modules in the SS-ISAR-Net plays a crucial role in determining the sparse ISAR imaging performance. A trade-off between computational efficiency and reconstruction accuracy should be made to ensure both high-quality imaging and practical feasibility. To investigate this, we systematically evaluate the impact of varying the number of network iteration layers, ranging from 2 to 14. These configurations are also employed for self-supervised training and subsequently tested on the dataset described in Sec.~\ref{Sec:4-2} with a down-sampling rate of 40$\,\%$. The corresponding results are presented in Fig.~\ref{fig:layer}. As illustrated in the figure, when the number of layers reaches $K \geq 12$, the performance gain becomes marginal. Therefore, $K=12$ is selected as the optimal configuration in this paper.

\begin{figure}[t]
  \hspace{-0.2cm}
  \includegraphics[width =8.5cm, height =6.2cm]{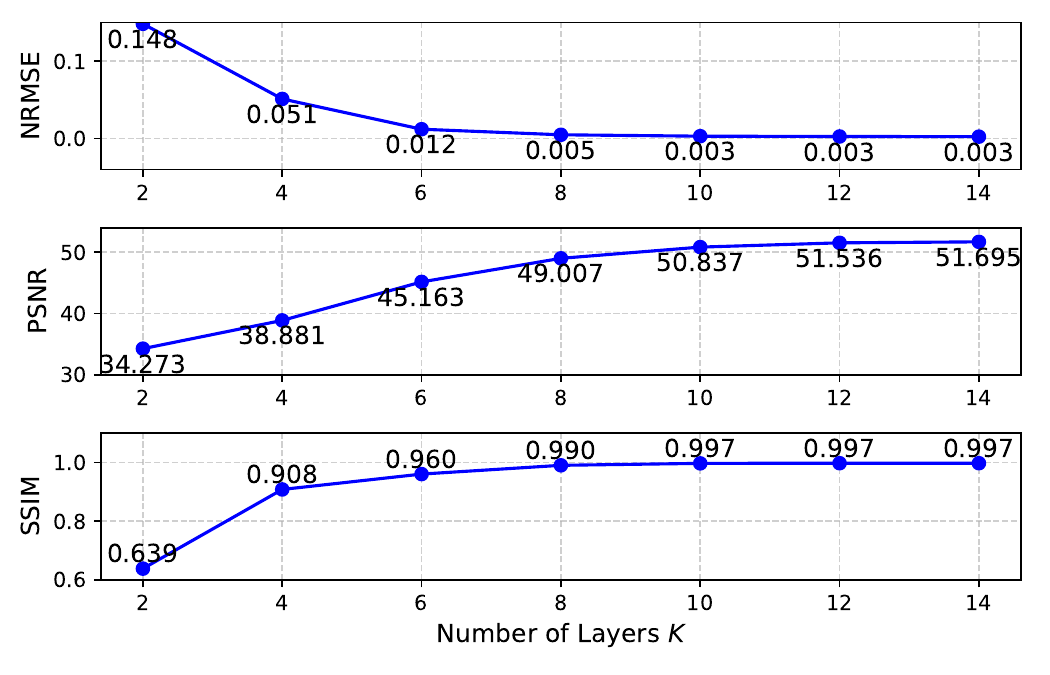}
  \caption{Effect of the iterative number $K$ of SS-ISAR-Net on the reconstruction performance, from top to bottom are NMSE, PSNR, and SSIM.}\label{fig:layer}
\end{figure}

\subsection{Kernel size of convolution layers in LFAT module}
\label{Sec:5-2}

An ablation study is conducted to analyze the impact of the kernel size of convolution layers in the proposed LFAT module on reconstruction performance.
Under the condition of a fixed number of convolutional layers, different kernel sizes influence the receptive field, thereby affecting the range of structural features' influence on threshold selection. A suitable size of convolution kernels can effectively learn the local features of the target. The experimental results are summarized in Table~\ref{table:Kernel size}. Note that a kernel size of 0 indicates the use of a fixed learning threshold across the entire image in~\eqref{eq:soft}. As shown in the results, the incorporation of the adaptive threshold module significantly enhances the accuracy of the reconstruction. A steady improvement is observed as the size of the kernels increases. However, beyond a kernel size of seven, the performance gain becomes negligible. At this point, the receptive field reaches $13\times13$, which is sufficient to capture key structural information. Therefore, a kernel size of seven is selected for the adaptive threshold module in the paper.

\begin{table}[t]
\centering
  \caption{Effect of the kernel size of convolution in LFAT module on the reconstruction performance (NMSE, PSNR, and SSIM)}   
  \resizebox{6.5cm}{1.3cm}{
  \begin{tabular}{cccc}
  \toprule
    kernel size & NMSE ($\downarrow $)    & PSNR ($\uparrow $)   & SSIM ($\uparrow $)   \\
    \midrule
    0 & 0.0276 &    43.4995 &   0.9664\\
    3 & 0.0152 &	47.0371 &	0.9775\\
    5 & 0.0034 &    48.5873 &   0.9966\\
    7 & 0.0027 &	50.1504 &	0.9970\\
    9 & 0.0027 &	50.0941 &	0.9973\\
  \bottomrule
  \end{tabular}}\label{table:Kernel size}
\end{table}

\subsection{EC regularization coefficient $\alpha$}\label{Sec:5-3}

For the proposed SS-ISAR-Net, which incorporates the equivariant constraint, selecting an appropriate regularization coefficient is crucial to ensure high reconstruction quality. In this study, we systematically evaluated the impact of different regularization coefficients, specifically 0.01, 0.1, 1, 10, 100, and 1000 on network performance. The SS-ISAR-Net is trained using these coefficients and subsequently tested on the dataset described in Sec.~\ref{Sec:4-2} with a down-sampling rate of 40$\,\%$. The quantitative results, summarized in Table~\ref{table:EC regularization} demonstrate that the regularization coefficient significantly influences test performance. Specifically, values within the range $\alpha \in [1,100]$ produce the best results, with $\alpha = 1$ selected in this paper.

\begin{table}[t]
\centering
  \caption{Effect of the EC regularization coefficient $\alpha$ on the reconstruction performance (NMSE, PSNR, and SSIM)}  
  \resizebox{6.5cm}{1.4cm}{
  \begin{tabular}{cccc}
  \toprule
    Coefficient $\alpha$ & NMSE ($\downarrow $)    & PSNR ($\uparrow $)   & SSIM ($\uparrow $)   \\
    \midrule
    0.01 & 0.4725  & 29.1548 & 0.2478\\
    0.1 &0.0152&	44.8870 &	0.9775	\\
    1 & 0.0034	&50.8274&0.9966\\
    10 &0.0041&	50.1504&	0.9963	 \\
    100&0.0041&	50.1011&	0.9960	 \\
    1000&0.0053&	48.5308&	0.9890	 \\
  \bottomrule
  \end{tabular}}\label{table:EC regularization}
\end{table}

\section{Conclusion}\label{Sec:6}
To reduce the need for paired ISAR images and echo, this article proposes a self-supervised learning-based unfolded ISAR imaging network (SS-ISAR-Net). By integrating an ADMM-based unfolded ISAR-Net with an equivariant constraint, the proposed method enables end-to-end training using only sparse radar echo data. Additionally, an N2N-based self-denoising module effectively mitigates the impact of noise, further enhancing the robustness of the proposed SS-ISAR-Net. Numerical simulations and real-data experiments demonstrate that the proposed SS-ISAR-Net obtains better ISAR images compared to other state-of-the-art methods. 

Finally, although two-dimensional (2D) sparse ISAR imaging was considered, the same idea used in SS-ISAR-Net could be, in principle, readily extended to 3D scenarios, such as distributed 3D radar, and would be conducted in the future.

\appendices

\section{Proof of the irrelevance of ${{\mathbf{N}}_{n1}}$ and ${{\mathbf{N}}_{n2}}$}\label{proof1}
According to~\eqref{Noise}, calculate the relevant covariance for ${{\mathbf{N}}_{n1}}$ and ${{\mathbf{N}}_{n2}}$
\begin{align}  
\begin{gathered}
  \sum  = \left[ {\begin{array}{*{20}{c}}
  {{{\mathbf{N}}_{n1}}} \\ 
  {{{\mathbf{N}}_{n2}}} 
\end{array}} \right]\left[ {{{\mathbf{N}}_{n1}}{\text{ }}{{\mathbf{N}}_{n2}}} \right] \hfill \\
   = \left[ {\begin{array}{*{20}{c}}
  {{\mathbf{N}} + {{\mathbf{N}}_1}} \\ 
  {{\mathbf{N}} - {{\mathbf{N}}_1}} 
\end{array}} \right]\left[ {{\mathbf{N}} + {{\mathbf{N}}_1}{\text{ }}{\mathbf{N}} + {{\mathbf{N}}_1}} \right] \hfill \\
   = \left[ \begin{gathered}
  \operatorname{cov} \left( {{\mathbf{N}} + {{\mathbf{N}}_1},{\mathbf{N}} + {{\mathbf{N}}_1}} \right){\text{ }}\operatorname{cov} \left( {{\mathbf{N}} + {{\mathbf{N}}_1},{\mathbf{N}} - {{\mathbf{N}}_1}} \right) \hfill \\
  \operatorname{cov} \left( {{\mathbf{N}} - {{\mathbf{N}}_1},{\mathbf{N}} + {{\mathbf{N}}_1}} \right){\text{ }}\operatorname{cov} \left( {{\mathbf{N}} - {{\mathbf{N}}_1},{\mathbf{N}} - {{\mathbf{N}}_1}} \right){\text{ }} \hfill \\ 
\end{gathered}  \right] \hfill \\
   = \left[ {\begin{array}{*{20}{c}}
  {2{\sigma ^2}{\mathbf{I}}}&{\mathbf{0}} \\ 
  {\mathbf{0}}&{2{\sigma ^2}{\mathbf{I}}} 
\end{array}} \right] \hfill \\ 
\end{gathered} 
\label{eq:irrelevance}
\end{align} 
where $\operatorname{cov}\left( { \cdot , \cdot } \right)$ represents the matrix-related operation. Since noises ${\mathbf{N}} \sim \mathcal{CN}\left( {0,{\sigma ^2}I} \right)$ and ${\mathbf{N}}_1 \sim \mathcal{CN}\left( {0,{\sigma ^2}I} \right)$ are uncorrelated, one can know $\operatorname{cov}\left( {{\bf{N}},{{\bf{N}}}} \right) = {\sigma ^2}{\bf{I}}$, $\operatorname{cov}\left( {{\bf{N}}_1,{{\bf{N}}_1}} \right) = {\sigma ^2}{\bf{I}}$, and $\operatorname{cov}\left( {{\bf{N}},{{\bf{N}}_1}} \right) ={\bf{0}}$. 

Building upon the above proof, we know that the two noises ${{\mathbf{N}}_{n1}}$ and ${{\mathbf{N}}_{n2}}$ are uncorrelated, which is the fundamental requirement for the N2N framework's validity in our self-denoising module.
\section{Proofs of loss equivalence}\label{proof2}

Following~\eqref{eq:derivation1}, using the irrelevance of ${{\mathbf{N}}_{n1}}$ and ${{\mathbf{N}}_{n2}}$ in \eqref{eq:irrelevance}, one can know
\begin{equation}
  \begin{gathered}
  {\mathbb{E}_{{{\mathbf{Y}}_{n\_1}},{{\mathbf{Y}}_{n\_2}}}}\left\{ {\left\| {{f_d}\left( {{{\mathbf{Y}}_{n\_1}}} \right) - {{\mathbf{Y}}_{n\_2}}} \right\|_F^2} \right\} \hfill \\
   = {\mathbb{E}_{{{\mathbf{N}}_{n1}},{{\mathbf{N}}_{n2}}}}\left\{ {\left\| {{f_d}\left( {{\mathbf{Y}} + {{\mathbf{N}}_{n1}}} \right) - {\mathbf{Y}} - {{\mathbf{N}}_{n2}}} \right\|_F^2} \right\} \hfill \\
   = {\mathbb{E}_{{{\mathbf{N}}_{n1}},{{\mathbf{N}}_{n2}}}}\left\{ \begin{gathered}
  \left\| {{f_d}\left( {{\mathbf{Y}} + {{\mathbf{N}}_{n1}}} \right) - {\mathbf{Y}}} \right\|_F^2 \hfill \\
   + 2{\mathbf{N}}_{n2}^H\left( {{f_d}\left( {{\mathbf{Y}} + {{\mathbf{N}}_{n1}}} \right) - {\mathbf{Y}}} \right) + \left\| {{{\mathbf{N}}_{n2}}} \right\|_F^2 \hfill \\ 
\end{gathered}  \right\} \hfill \\
   = {\mathbb{E}_{{{\mathbf{N}}_{n1}},{{\mathbf{N}}_{n2}}}}\left\{ {\left\| {{f_d}\left( {{\mathbf{Y}} + {{\mathbf{N}}_{n1}}} \right) - {\mathbf{Y}}} \right\|_F^2} \right\} + const \hfill \\ 
\end{gathered} 
\end{equation}
where the third step of the derivation uses the condition that ${{\mathbf{N}}_{n2}}$ is independent of ${\mathbf{Y}}$ and ${{\mathbf{N}}_{n1}}$, and $\left\| {{{\mathbf{N}}_{n2}}} \right\|_F^2$ is a constant. Along the same lines, we know that
\begin{equation}
  \begin{gathered}
  {\mathbb{E}_{{{\mathbf{Y}}_{n\_1}},{{\mathbf{Y}}_{n\_2}}}}\left\{ {\left\| {{\mathbf{A}}{f_u}\left( {{f_d}\left( {{{\mathbf{Y}}_{n\_1}}} \right)} \right) - {{\mathbf{Y}}_{n\_2}}} \right\|_F^2} \right\} \hfill \\
   = {\mathbb{E}_{{{\mathbf{N}}_{n1}},{{\mathbf{N}}_{n2}}}}\left\{ {\left\| {{\mathbf{A}}{f_u}\left( {{f_d}\left( {{\mathbf{Y}} + {{\mathbf{N}}_{n1}}} \right)} \right) - {\mathbf{Y}} - {{\mathbf{N}}_{n2}}} \right\|_F^2} \right\} \hfill \\
   = {\mathbb{E}_{{{\mathbf{N}}_{n1}},{{\mathbf{N}}_{n2}}}}\left\{ \begin{gathered}
  \left\| {{\mathbf{A}}{f_u}\left( {{f_d}\left( {{\mathbf{Y}} + {{\mathbf{N}}_{n1}}} \right)} \right) - {\mathbf{Y}}} \right\|_F^2+ \left\| {{{\mathbf{N}}_{n2}}} \right\|_F^2 \hfill \\
   + 2{\mathbf{N}}_{n2}^H\left( {{\mathbf{A}}{f_u}\left( {{f_d}\left( {{\mathbf{Y}} + {{\mathbf{N}}_{n1}}} \right)} \right) - {\mathbf{Y}}} \right)  \hfill \\ 
\end{gathered}  \right\} \hfill \\
   = {\mathbb{E}_{{{\mathbf{N}}_{n1}},{{\mathbf{N}}_{n2}}}}\left\{ {\left\| {{\mathbf{A}}\left( {{f_u}\left( {{f_d}\left( {{\mathbf{Y}} + {{\mathbf{N}}_{n1}}} \right)} \right) - {\mathbf{X}}} \right)} \right\|_F^2} \right\} + const \hfill \\ 
\end{gathered}
\end{equation}
This can also be equated to loss training for the noiseless radar echo $\mathbf{Y}$.

\ifCLASSOPTIONcaptionsoff
  \newpage
\fi

\bibliography{ref}
\bibliographystyle{IEEEtran}

\begin{IEEEbiography}[{\includegraphics[width=1in,height=1.25in,clip,keepaspectratio]{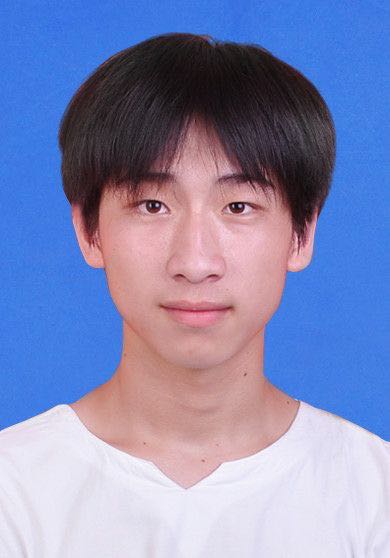}}]
{Ziwen Wang}{\space} was born in Jiangxi, China, in 1999. He received the B.Sc. degree in communication engineering from the Northwestern Polytechnical University, Xi’an, China, in 2020, and the M.Sc. degree in Information and Communication Engineering from the Beijing Institute of Technology, Beijing, in 2023. He is currently pursuing the Ph.D. degree with the Beijing Institute of Technology, Beijing.

His research interests mainly include radar signal processing, sparse aperture ISAR imaging, and deep learning.
\end{IEEEbiography}%

\begin{IEEEbiography}
[{\includegraphics[width=1in,height=1.25in,clip,keepaspectratio]{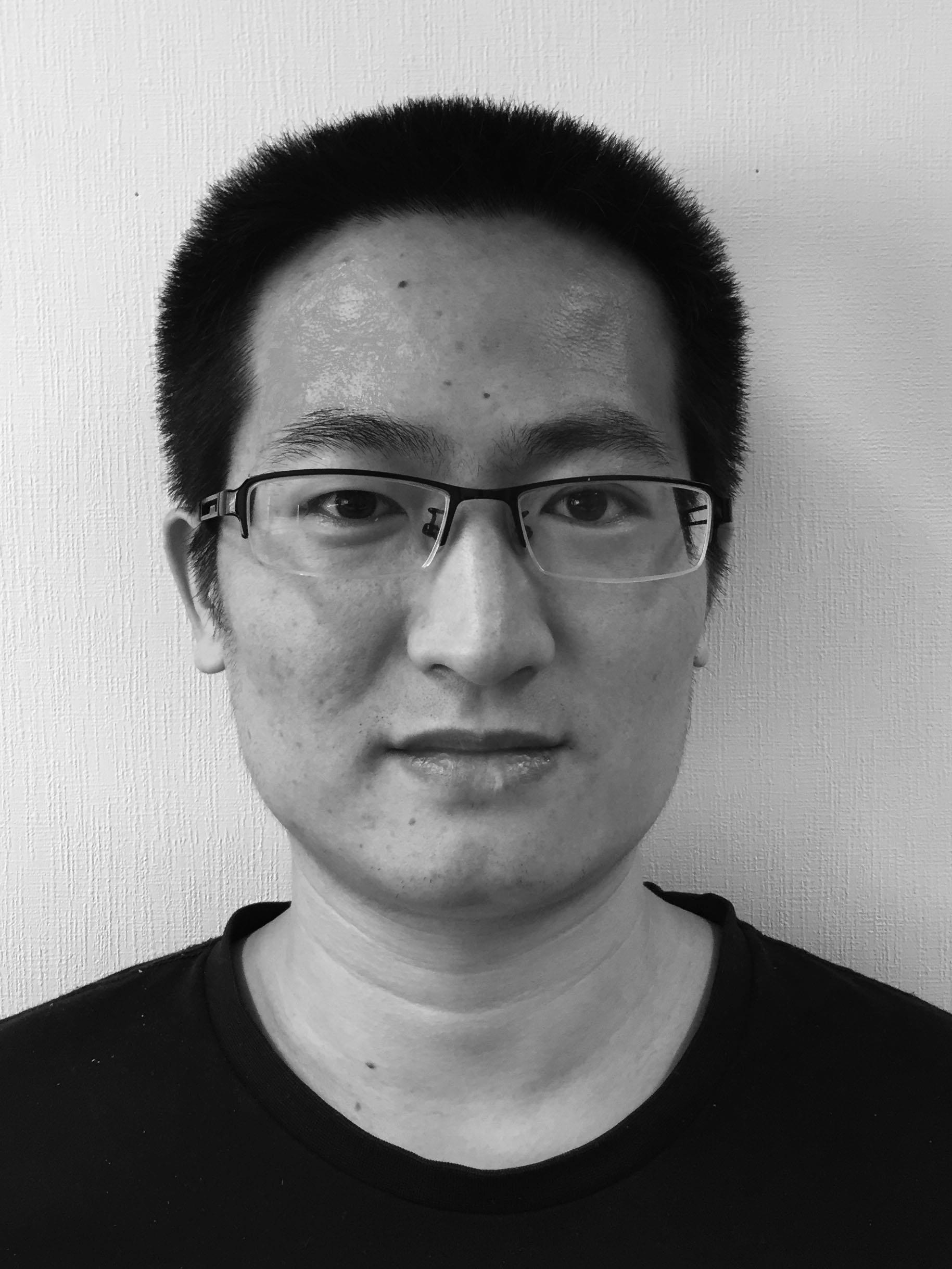}}]
{Jianping Wang} (Member, IEEE) received the Ph.D. degree in electrical engineering from the Delft University of Technology, Delft, The Netherlands, in 2018. From August 2012 to April 2013, he was a Research Associate at the University of New South Wales, Sydney, NSW, Australia, with a focus on frequency-modulated continuous wave synthetic aperture radar signal processing for formation flying satellites. He was a Post-Doctoral Researcher and a Guest Researcher with the Group of Microwave Sensing, Signals and Systems (MS3), Delft University of Technology, from 2018 to 2024. Since 2024, he has been with the School of Information and Electronics, Beijing Institute of Technology, Beijing, China. 

His research interests include microwave imaging, signal processing, and antenna array design. Dr. Wang was a TPC Member of the IET International Radar Conference, Nanjing, China, in 2018 and 2023. He was a Finalist for the Best Student Paper Award in the International Workshop on Advanced Ground Penetrating Radar (IWAGPR), Edinburgh, U.K., in 2017, and the International Conference on Radar, Brisbane, Australia, in 2018. He has served as a reviewer of many IEEE journals
\end{IEEEbiography}

\begin{IEEEbiography}
[{\includegraphics[width=1in,height=1.25in,clip,keepaspectratio]{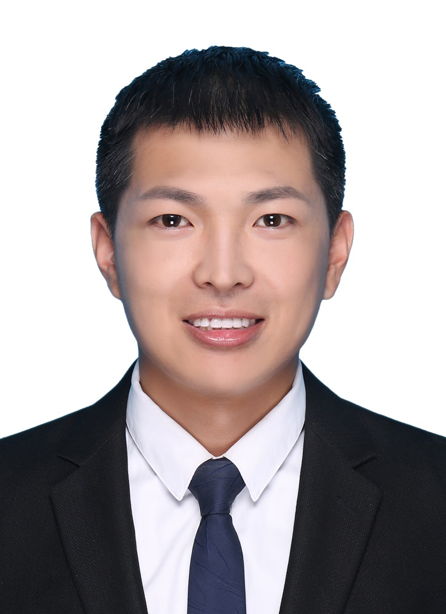}}]
{Pucheng Li} (Member, IEEE) received the B.Eng. and M.Eng. degrees in electronic engineering from the Civil Aviation University of China, Tianjin, China, in 2017 and 2020, respectively, and the Ph.D. degree from the Beijing Institute of Technology (BIT), Beijing, China, in 2024. 

He is currently a Postdoctoral Fellow with the Radar Technology Research Institute at BIT. His research focuses on high-resolution radar imaging and machine learning.
\end{IEEEbiography}

\begin{IEEEbiography}
[{\includegraphics[width=1in,height=1.25in,clip,keepaspectratio]{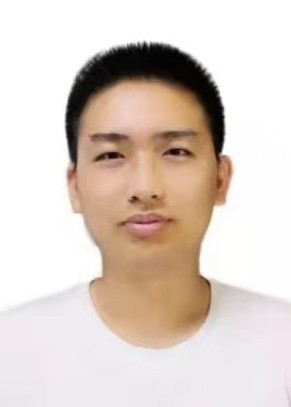}}]
{Yifan Wu}{\space} was born in Henan, China, in 1999. He received his B.Sc. degree in remote sensing science and technology from Xidian University, Xi'an, China, in 2021. He is currently pursuing the Ph.D. degree with the Beijing Institute of Technology, Beijing. 

His research interests mainly include radar signal processing, radar target recognition, and deep learning.
\end{IEEEbiography}

\begin{IEEEbiography}
[{\includegraphics[width=1in,height=1.25in,clip,keepaspectratio]{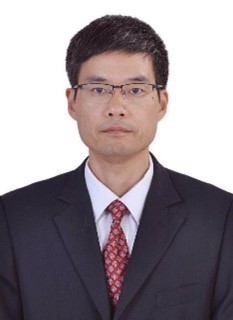}}]
{Zegang Ding} (Senior Member, IEEE) received the Ph.D. degree from the Beijing Institute of Technology (BIT), Beijing, China, in 2008. In 2008, he was a Visiting Scholar with the Communication Group, School of Electronic and Electrical Engineering, University of Birmingham, Birmingham, U.K. He is currently a part of the Teaching Staff at the Department of Electronic Engineering, BIT. 

His research interests include synthetic aperture radar imaging and system design. 
\end{IEEEbiography}
\end{document}